\documentclass[onecollarge]{svjour2}                    
\usepackage{graphicx}
\usepackage{amsmath,amsfonts,amssymb,amsopn,amscd}
\usepackage{bbm,wasysym,stmaryrd}
\usepackage{subfigure}
\usepackage{color,colortbl}
\usepackage{hyperref}

\newcommand{\Exp}{\mathbbm{E}}
\newcommand{\Ex}[1]{\mathbbm{E}\left [#1 \right]} 

\newcommand{\E}{\mathcal{E}}
\DeclareMathOperator{\eqlaw}{\stackrel{\mathcal{L}}{=}}

\newcommand{\erf}{\text{erf}}

\newcommand{\R}{\mathbbm{R}}
\newcommand{\N}{\mathbbm{N}}

\newcommand{\F}{\mathcal{F}}

\newcommand{\M}{\mathcal{M}}
\newcommand{\C}{\mathcal{C}}
\newcommand{\Ct}{\mathcal{C}_{\tau}}

\graphicspath{{./Figures/}}
\journalname{Journal of Statistical Physics}
\smartqed

\begin{document}
	\title{Limits and dynamics of stochastic neuronal networks with random heterogeneous delays \thanks{CNRS UMR 7241m INSERM U1050, Universit\'e Pierre et Marie Curie ED 158. MEMOLIFE Laboratory of excellence and Paris Sciences Lettres PSL*.}
	}

	\author{Jonathan Touboul}

	\institute{ The Mathematical Neuroscience Laboratory \\
				Coll\`ege de France / CIRB and INRIA Bang Laboratory\\
				11, place Marcelin Berthelot,
				75005 Paris, France
	            Tel.: +33-144271388\\
	            \email{jonathan.touboul@college-de-france.fr}           
	}

	\date{Received: date / Accepted: date}


%
%
%
%

\maketitle

\begin{abstract}
Realistic networks display a heterogeneous distribution of transmission delays. We analyze here the limits of large stochastic multi-populations networks with stochastic coupling and random interconnection delays. We show that depending on the nature of the delays distributions, a quenched or averaged propagation of chaos takes place in these networks, and that the network equations converge towards a delayed McKean-Vlasov equation with distributed delays. Our approach is mostly developed for neuroscience applications. We instantiate in particular a classical neuronal model, the Wilson and Cowan system, and show that the obtained limit equations have Gaussian solutions whose mean and standard deviation satisfy a closed set of coupled delay differential equations in which the distribution of delays and the noise levels appear as parameters. This allows to uncover precisely the effects of noise, delays and coupling on the dynamics of such heterogeneous networks, in particular their role in the emergence of synchronized oscillations. We show in several examples that not only the averaged delay, but also the dispersion, govern the dynamics of such networks.
\end{abstract}






\section*{Introduction}
Realistic networks show heterogeneous interconnection delays. The cortex is a paradigmatic example of such networks. In the brain, neurons form large-scale assemblies (termed neural populations) subject to an intense noise. They communicate through the emission of spikes (action potentials) that are transmitted through synapses that can be chemical or, less often, electrical. The transport of spikes through the axons takes a specific time, function of the length of the axon and of the time taken for the signal to be transmitted through synapses, and therefore are highly heterogeneous. Moreover, chemical synapses modify the membrane potential of the post-synaptic neuron through a complex mechanism of release and binding of neurotransmitters. This mechanism takes a specific time, and results in random variations of the current transmitted through the different synapses and also from spike to spike.

Models describing the emergent behavior arising from the interaction of neurons in large-scale networks largely overlooked both phenomena of heterogeneous delays and stochastic synapses. Most models have relied on continuum limits ever since the seminal work of Wilson and Cowan and Amari \cite{amari:72,amari:77,wilson-cowan:72,wilson-cowan:73}. Such models tend to represent the activity of the network through a macroscopic variable, the population-averaged firing rate, that is generally assumed to be deterministic. Many analytical and numerical results and properties have been derived from these equations and related to cortical phenomena, for instance for the problem of  spatio-temporal pattern formation in spatially extended models (see e.g.~\cite{coombes-owen:05,ermentrout:98,ermentrout-cowan:79,laing-troy-etal:02}). This approach tends to implicitly make the assumption that the effect of noise and heterogeneities vanishes in large populations. A different approach has been to study regimes where the activity is uncorrelated. A number of computational studies on the integrate-and-fire neuron showed that under certain conditions neurons in large assemblies end up firing asynchronously, producing null correlations \cite{abbott-van-vreeswijk:93,amit-brunel:97,brunel-hakim:99}. In these regimes, the correlations in the firing activity decrease towards zero in the limit where the number of neurons tends to infinity. The emergent global activity of the population in this limit is deterministic, and evolves according to a mean-field firing rate equation. However, these states arise in sparsely connected networks which is not the case of cortical columns. In order to go beyond asynchronous states and take into account the stochastic nature of the firing and how this activity scales as the network size increases, different approaches have been developed, such as the population density method and related approaches \cite{cai-tao-etal:04}. Most of these approaches involve expansions in terms of the moments of the resulting random variable, and the moment hierarchy needs to be truncated which is a quite hard task that can raise a number of technical issues (see e.g.\cite{ly-tranchina:07}). Homogeneous spatially extended systems of neurons were recently analyzed using the theory of probability~\cite{touboulNeuralfields:11,touboulNeuralFieldsDynamics:11}. These were the first studies considering individual delays between neurons. In these studies, the propagation delays between neurons were considered constant, and in these cases, no specific delay-induced transition was identified. Heterogeneous networks have been considerably less analyzed. One notable exception is the work of Sompolinsky and collaborators~\cite{sompolinsky-crisanti-etal:88} recently analyzed from a mathematical viewpoint using large deviations techniques~\cite{cabana-touboul:12}, where the authors considered random heterogeneous connections between neurons.

The main question we shall address in this  article is the effect of heterogeneous delays on the macroscopic dynamics of such neuronal areas in the presence of stochastic coupling. Although it now firmly established that delays are fundamental for shaping the activity of large-scale neuronal networks~\cite{roxin-brunel-etal:05,coombes-laing:11,roxin-montbrio:11,series-fregnac:02}, very little is understood about the type of delays involved in such networks. The question of how these individual delays translate in macroscopic limits, and how the heterogeneity in the distribution of delays affect macroscopic behaviors, hence has deep implications, and is still largely open. At the microscopic scale, it is relatively clear that transmission delays are well approximated by constant delay depending on the distance between neurons and on the typical synaptic time scale. These constant delays hence vary depending on the pair of neurons considered and are related to the topology of the network in the physical space.

Considering random heterogeneous delays between neurons, we establish the limit, as the number of neurons tend to infinity, of the network equations. We will show that the propagation of chaos property occurs in these systems, and that the asymptotic macroscopic behavior in the limit where the number of neurons tends to infinity is described by a McKean-Vlasov equation with distributed delays, the delay measure being exactly the statistical distribution of individual delays.
The dynamics of the obtained McKean-Vlasov equations with distributed delays is very hard to analyze for general neuron models. In order to uncover the macroscopic dynamics of the neurons and their dependence on the distribution of delays, we eventually instantiate a particularly suitable model, the classical firing-rate model. In that case, that solutions are Gaussian, and the mean and standard deviation satisfy a closed system of delay differential equations with distributed delays. This reduction allows using classical tools of the domain of infinite-dimensional dynamical systems to show generic qualitative transitions between different qualitative states, as a function of the delays and of the variance of the synaptic weights.

\section{Setting of the problem}\label{sec:Setting}

In this section, we describe the mathematical formalism used throughout the manuscript. We shall analyze the dynamics of a relatively general neuron model, valid for most usual models used in computational neuroscience such as the Hodgkin-Huxley~\cite{hodgkin-huxley:52} or Fitzhugh-Nagumo~\cite{fitzhugh:69} models. The state of each neuron $i$ in the network is classically described by a $d$-dimensional variable $X^i\in E:=\R^d$, typically corresponding to the membrane potential of the neuron and possibly additional variables such as those related to ionic concentrations and gated channels (see e.g.~\cite{ermentrout-terman:10b}).

We consider networks composed of a finite number of populations $\alpha \in \{1, \cdots, P\}$. The populations differ through the intrinsic properties of each neurons and the input they receive. We assume that the number of neuron in each population $\alpha \in \{1,\ldots,P\}$, denoted $N_{\alpha}$ increases as the network size increases, and tends to infinity in the limit considered, which we will be denoting $N\to\infty$. We define the population function that associates to a neuron $i$ the population $\alpha$ it belongs to: $p(i)=\alpha$.

Let us consider a single neuron $i$ belonging to population $\alpha$. The state of the neuron, $X^i$, has an intrinsic dynamics governed by a nonlinear drift function $f_{\alpha}:\R\times E \mapsto E$ (including the intrinsic dynamics and the deterministic inputs) and a diffusion matrix $g_{\alpha}:\R\times E \mapsto E^{m}$. Each individual process $(X^i_t)$ takes values in $E$, and when there is no interaction, is governed by the equation:
\[dX^{i}_t= f_{\alpha}(t,X^{i}_t)\, dt + g_{\alpha}(t,X^{i}_t)\,dW_t^i.\]
where $(W^{i}_t)_{i\in\N}$ are independent $m\times d$ Brownian motions\footnote{The degree of freedom given by the choice of $m$-dimensional Brownian motions allows to take into account different sources of independent sources of noise, typically in conductance-based neurons thermal noise affecting the voltage and channel noise affecting ionic concentrations.}

This equation corresponds to the intrinsic dynamics of each individual diffusion, i.e. the behavior of a single neuron when disconnected of the network. When included in the network, the neurons interact with all other neurons and this interaction is modeled through a continuous function $b_{\alpha\gamma}(x,y): E \times E \mapsto E$, multiplied by a coefficient, called the synaptic weight $w_{ij}$. These synaptic weights are themselves stochastic, due to the nature of the synaptic transmission. Incoming interactions are known to be bounded, the total incoming connectivity from a given population (say $\gamma$) to a single neuron from population $\alpha$ has a fixed averaged $J_{\alpha\gamma}$ and standard deviation $\sigma_{\alpha\gamma}$. We consider that the synaptic weights are stochastic, and are assumed to be equal, for a neuron $i$ in population $\alpha$ and a neuron $j$ in population $\gamma$, to
\[w_{ij}=w_{i\gamma}=\frac 1 {N_{\gamma}} J_{\alpha\gamma} + \frac{\sigma_{\alpha\gamma}}{N_{\gamma}} \;  \dot{B}^{i\gamma}\]
where $\dot{B}^{i\gamma}$ are independent $\delta$-dimensional ``white noise'' processes related to the independent Brownian motions $B^{i,\gamma}$ introduced\footnote{Note that this model has the disadvantage to possibly change the excitatory or inhibitory nature of the connection from neuron $i$ to neuron $j$ when the sign of $w_{ij}$ changes. Other more biologically relevant problems involve enforcing the synaptic weight not to change sign by modeling it as the solution of particular stochastic differential equation that does not change sign, such as the Cox-Ingersoll-Ross process, and by fit our framework by formally adding the $P$ synaptic weights of neuron $i$ in the state variable $X^i$ of neuron $i$. }.

The delay $\tau_{ij}$ occurring in the interaction between neuron $j$ and $i$, and related to the information transmission through the axons and the synapses, depend on the precise distance between neurons $i$ and $j$. We assume that for any fixed neurons $i\in\{1\cdots N\}$, the delays $(\tau_{ij})_{j\neq i}$ are pairwise independent (and independent of the Brownian motions involved in the network dynamics) and identically distributed in each population, with a law denoted $\eta_{i\,p(j)}$. This is a reasonable assumption from the modeling viewpoint. Indeed, neural populations can be considered to form certain clusters in the space in which the specific locations of neurons can be considered independent, and delays hence depend on the typical distance between the populations, their dispersion, and on the characteristic time constants of the synapse. From the referential of one given neuron $i$, the resulting delays are hence independent. From the global network viewpoint, these delays are of course correlated (distance between two points are symmetrical, topological relationship between neuronal locations induces relationship between distances, e.g. triangular inequality). Let us denote by $\tau$ the largest possible delay in the network. All delay distributions hence charge only the interval $[-\tau,0]$. We eventually denote by $\eta_{\alpha\gamma}$ the averaged delay between neurons of population $\alpha$ and neurons of population $\gamma$, were the average is taken over all possible values of $\eta_{i\gamma}$ for $i$ in population $\alpha$ (for instance, averaged on all possible locations for neurons of population $\alpha$).

The infinitesimal current the neuron $i$ receives from its neighbors is hence modeled as:
\begin{equation*}
	\sum_{\gamma=1}^P \sum_{p(j)=\gamma} \frac{J_{\alpha\gamma}}{N_{\gamma}} b_{\alpha\gamma}(X^{i,N}_t,X^{j,N}_{t-\tau_{ij}}) \, dt
	+\sum_{\gamma=1}^P \frac{\sigma_{\alpha\gamma}}{N_{\gamma}}\sum_{p(j)=\gamma} \beta_{\alpha\gamma}(X^{i,N}_t,X^{j,N}_{t-\tau_{ij}}) \, dB^{i\gamma}_t
\end{equation*}

For the sake of simplicity, we make an abuse of notations and replace the notation $J_{\alpha\gamma}b_{\alpha\gamma}(x,y)$ (resp. $\sigma_{\alpha\gamma}\beta_{\alpha\gamma}(x,y)$) by $b_{\alpha\gamma}(x,y)$ (resp. $\beta_{\alpha\gamma}(x,y)$). The resulting network equation reads:
\begin{multline}\label{eq:Network}
	d\,X^{i,N}_t=f_{\alpha}(t,X^{i,N}_t) \, dt + \sum_{\gamma=1}^P \sum_{j=1 , p(j)=\gamma}^{N} \frac{1}{N_{\gamma}} b_{\alpha\gamma}(X^{i,N}_t,X^{j,N}_{t-\tau_{ij}}) \,dt \\
	+ g_{\alpha}(t,X^{i,N}_t)\cdot dW^{i}_t
	+ \sum_{\gamma=1}^P \frac{1}{N_{\gamma}} \sum_{j=1 , p(j)=\gamma}^{N}  \beta_{\alpha\gamma}(X^{i,N}_t,X^{j,N}_{t-\tau_{ij}}) \cdot dB^{i\gamma}_t.
\end{multline}
These equations are stochastic differential equations on the infinite-dimensional space of functions $C([-\tau,0],E)$ (i.e. on the variable $\tilde{X}_t=(X_s,s\in[t-\tau,t])$, see e.g.~\cite{da-prato:92,mao:08}) of continuous functions of $[-\tau,0]$ with values in $E$. Let us denote by $\Ct=C([-\tau,0],E^P)$. The initial conditions are not specified at this point. We will sometimes make the assumption that the network has chaotic initial condition, i.e. independent identically distributed initial conditions. In details, for $(\zeta_0^{\alpha}(t), \alpha=1\cdots P) \in \Ct$ a stochastic process with independent components, a chaotic initial condition on the network consists in setting the initial condition of all neurons $i$ of population $\alpha$ to an independent copy of $\zeta^{\alpha}_0$.

For any $(\alpha,\gamma)\in\{1,\cdots,P\}^2$, we make the following technical assumptions on the parameters of the model:
\renewcommand{\theenumi}{(H\arabic{enumi})}
\begin{enumerate}
	\item\label{Assump:LocLipsch} $f_{\alpha}$ and $g_{\alpha}$ are uniformly $K$-Lipschitz-continuous with respect to their second variable
	\item\label{Assump:LocLipschb} $b_{\alpha\gamma}$ and $\beta_{\alpha\gamma}$ are $L$-Lipschitz-continuous, i.e. there exists a positive constant $L$ such that for all $(x,y)$ and $(x',y')$ in $E \times E$ we have:
\[\vert \Theta_{\alpha\gamma}(x,y)-\Theta_{\alpha\gamma}(x',y') \vert  \leq L \, (\vert x-x'\vert + \vert y-y'\vert ) \]
for $\Theta\in\{b,\,\beta\}$.

\item\label{Assump:bBound} There exists a $\tilde{K}>0$ such that:
\[\max(\vert b_{\alpha\gamma}(x,z)\vert^2,\vert \beta_{\alpha\gamma}(x,z)\vert^2) \leq \tilde{K}\]
\item\label{Assump:MonotoneGrowth} The drift and diffusion functions satisfy the monotone growth condition:
	\begin{equation*}
		x^T f_{\alpha}(t,x) + \frac 1 2 \vert g_{\alpha}(t,x) \vert^2  \leq K \; (1+\vert x \vert^2).
	\end{equation*}
\end{enumerate}
\begin{remark}
	In a previous version of this paper and in the published text, it was indicated that the drift and diffusions can be assumed locally Lipschitz continuous, to extend to the particular case of the Fitzhugh-Nagumo neuron model~\cite{fitzhugh:55}. All developments were nevertheless done under assumption (H1) above, since most classical neuron models satisfy that property. The extension to locally Lipschitz-continuous diffusions hinted in the previous version would actually require to develop finer techniques that those suggested, in the flavor of the methodology developed for FitzHugh-Nagumo neural networks in~\cite{quininao}, to ensure existence and uniqueness of solutions of the limit equation.
\end{remark}

In what follows, $\M^2(C([-\tau,0], E^N)$ denotes the space of square integrable stochastic processes on $[-\tau, 0]$ with values in $E^N$. Let us first state the following proposition ensuring well-posedness of the network system under the assumptions of the section:
\begin{proposition}\label{pro:ExistenceUniquenessNetwork}
	Let $X_0 \in \M^2(C([-\tau,0], E^N))$ an initial condition of the network system. Under the assumptions of the section, there exists a unique strong solution to the network equations \eqref{eq:Network}.
\end{proposition}

The proof of this proposition is a direct application of Mao's result~\cite{mao:02}, and essentially uses the same arguments as those of the proof theorem~\ref{thm:ExistenceUniqueness}. The interested reader is invited to follow the steps of this demonstration to prove proposition~\ref{pro:ExistenceUniquenessNetwork}.

We address the problem of the asymptotic behavior, as the number of neurons $N$ goes to infinity, of these network equations, in particular the question of the propagation of chaos property in this randomly delayed, randomly interacting system. Assuming translation invariance in the neural field, the delays distribution $\eta_{ij}$ is only function of the neural population $i$ belongs to. In that case, for any realization of the synaptic delays $\tau_{ij}$, the propagation of chaos property states that, provided that the initial condition on the network is chaotic (i.e. for all $i$ in population $\alpha$, the initial condition $X^i_t\vert_{[-\tau,0]}$ is an independent copy of a given process $\zeta^{\alpha} \in \Ct$) and when $N\to \infty$, the state of each neurons are independent for all times, and have a law given by a nonlinear McKean-Vlasov equation\footnote{McKean-Vlasov is the generic term in the mathematical literature on statistical physics for describing the type of equation we obtain. It is differs from the well-know to physicists Vlasov equation, which describes the evolution of the distribution function of plasma consisting of charged particles with long-range interactions, in particular in that our equation includes a diffusion term (Brownian motion) and is written as a stochastic equation, whereas usual the Vlasov equation is written as a nonlinear PDE closer from what we call the \emph{McKean-Vlasov-Fokker-Planck} equation.} that only depends on the neural population (say, $\alpha$), which is the unique solution of the McKean-Vlasov equations with distributed delays:
\begin{multline}\label{eq:MFE}
	d\bar{X}^{\alpha}_t= f_{\alpha}(t,\bar{X}^{\alpha}_t)\,dt + \sum_{\gamma=1}^P \int_{-\tau}^0  \Exp_{\bar{Y}}[b_{\alpha\beta}(\bar{X}^{\alpha}_t, \bar{Y}^{\beta}_{t+s})]d\eta_{\alpha\beta}(s) \,dt \\
	+ g_{\alpha}(t,\bar{X}^{\alpha}_t)\,dW_t^{\alpha}
	+ \sum_{\gamma=1}^P\int_{-\tau}^0 \Exp_{\bar{Y}}[\beta_{\alpha\beta}(\bar{X}^{\alpha}_t, \bar{Y}^{\beta}_{t+s})]d\eta_{\alpha\beta}(s) \,dB^{\alpha\gamma}_t
\end{multline}
where $\bar{Y}$ is a process independent of $\bar{X}$ that has the same law, $\Exp_{\bar{Y}}$ the expectation under the law of $\bar{Y}$, $W^{\alpha}_t$ and $B^{\alpha\gamma}_t$ are independent adapted standard Brownian motions of dimensions $m\times d$ and $\delta \times d$ respectively. In other words, the mean field equation can be written, denoting by $m_t^{\gamma}(dx)$ the law of $\bar{X}^{\gamma}_t$:
\begin{multline*}
	d\bar{X}^{\alpha}_t= f_{\alpha}(t,\bar{X}^{\alpha}_t)\,dt + \sum_{\gamma=1}^P J_{\alpha\gamma}\int_{-\tau}^0\int_E b_{\alpha\gamma}(\bar{X}^{\alpha}_t, y) m^{\gamma}_{t+s}(dy)\,d\eta_{\alpha\gamma}(s) dt \\
	+ g_{\alpha}(t,\bar{X}^{\alpha}_t)\,dW_t^{\alpha}
	+ \sum_{\gamma=1}^PJ_{\sigma\gamma}\int_{-\tau}^0\int_E \beta_{\alpha\gamma}(\bar{X}^{\alpha}_t, y) m^{\gamma}_{t+s}(dy)\,d\eta_{\alpha\gamma}(s) dB^{\alpha\gamma}_t
\end{multline*}
where $m^{\alpha}_t$ is the probability distribution of $\bar{X}^{\alpha}_t$. In these equations, $W_t^{\alpha}$ for $\alpha=1\cdots P$ are independent standard $m$-dimensional Brownian motions. This equation is also classically written as the McKean-Vlasov Fokker-Planck equation on the probability density, with a slightly abuse of notations denoted $m^{\gamma}_{t}(y)$, of the solution:
\begin{multline}\label{eq:McKeanVlasovFP}
	\partial_t m^{\alpha}_t(x)= -\nabla_x \Big(f_{\alpha}(t,x)\;m^{\alpha}_t(x)\Big) -\nabla_x \Big(\sum_{\gamma=1}^P \int_{-\tau}^0 \int_E b_{\alpha\gamma}(x, y) m^{\gamma}_{t+s}(y)\,dy\,d\eta_{\alpha\gamma}(s) m^{\alpha}_t(x)\Big) \\
	+ \frac 1 2 \Delta_x \Big(g_{\alpha}(t,x)^2 m^{\alpha}_t(x)\Big)
	+ \frac 1 2 \Delta_x  \Big\{\sum_{\gamma=1}^P\Big(\int_{-\tau}^0 \int_E \beta_{\alpha\gamma}(x, y) m^{\gamma}_{t+s}(y)\,dy\,d\eta_{\alpha\gamma}(s)\Big)^2 m^{\alpha}_t(x)\Big\}
\end{multline}
with initial condition $m^{\alpha}_t(x)\vert_{t\in [-\tau,0]}=p^{\alpha}(t,x)$ where $p^{\alpha}$ is the density of $\zeta^{\alpha}_0$, if it exists.

The same result holds when the neural field is not translation invariant (i.e. $\eta_{i\beta}$ varies depending on the choice of neuron $i$ in population $\alpha$), but in that case the result holds for the law of the limit distribution averaged on the law of the delays.

In that view, the equation \eqref{eq:MFE} is an implicit equation on the law of $\bar{X}_t$ (denoting when no confusion is possible the vector $(\bar{X}^{\alpha}_t,\alpha=1\cdots P)$). In order to prove the propagation of chaos property, we will use the coupling method (see e.g.~\cite{sznitman:89}). The proof is in two steps: (i) we prove that the equation \eqref{eq:MFE} has a unique solution, and (ii) that the law of $X^{i,N}_t$ converges towards the law of \eqref{eq:MFE} and more precisely that for any finite set of neurons $\{i_1,\cdots,i_k\}$ the law of $(X^{i_1,N}_t,\cdots,X^{i_k,N}_t, t\in[-\tau, T])$ converge almost surely towards a vector $(\bar{X}^{1}_t, \cdots, \bar{X}^{k}_t, t\in[-\tau,T])$ where the processes $\bar{X}^{l}$ are independent and have the law of $X^{p(i_l)}$ given by \eqref{eq:MFE}.

\section{Well-Posedness of the delayed McKean-Vlasov equations}

In this section, we show existence and uniqueness of solutions for the delayed McKean-Vlasov equation~\eqref{eq:MFE} that constitute the putative mean-field limit. Let us denote by $\mathcal{M}(\mathcal{C})$ the set of probability distributions on $\mathcal{C}$ the set continuous functions $[-\tau,T] \mapsto E^{P}$, and $\M^2(\C)$ the space of square-integrable processes.

Let $(W^{\alpha}\;;\; \alpha =1\cdots P)$ (respectively $(B^{\alpha\gamma}, \;;\; \alpha,\gamma =1\cdots P)$) be a family of $P$ (resp. $P^2$) independent $m\times d$ (resp. $\delta\times d$)-dimensional standard Brownian motion adapted on $(\Omega,\F,P)$. Let us also fixe $\zeta_0 \in \M^2(\Ct)$ a random variable, and $(\zeta^i_0)_{i\in \N}$ an iid sequence with the same law as $\zeta_0^{p(i)}$, that constitutes the chaotic initial condition of the network equations. We introduce the map $\Phi$ map acting on stochastic processes and defined by:
\[\Phi:
\begin{cases}
	\mathcal{M}(\mathcal{C}) &\mapsto \mathcal{M}(\mathcal{C})\\
	X &\mapsto (Y_t=\{Y^{\alpha}_t, \alpha=1\cdots P\})_t \text{ with } \\
	&Y_t^{\alpha}=\zeta_0^{\alpha}(0) + \int_0^t \Big(f_{\alpha}(s,X_s^{\alpha}) + \sum_{\gamma=1}^P \int_{-\tau}^0\Exp_{Z}[b_{\alpha\gamma}(X_s^{\alpha}, Z^{\gamma}_{s+u})] \Big)\,d\eta_{\alpha\gamma}(u) \,ds \\
	& + \int_0^t g_{\alpha}(s,X^{\alpha}_s)\,dW_s^{\alpha} + \sum_{\gamma=1}^P \int_0^t \int_{-\tau}^0\Exp_{Z}[\beta_{\alpha\gamma}(X_s^{\alpha}, Z^{\gamma}_{s+u})]\,d\eta_{\alpha\gamma}(u) \cdot dB_s^{\alpha\gamma} \;;\; t>0 \\
	& \\
	& Y_t=\zeta_0^{\alpha}(t) \qquad t\in [-\tau, 0]
\end{cases}
\]
where the process introduce a process $Z_t$ has the same law as $X$ and to be independent of $X$. There is a trivial identification between the solutions of the mean-field equation \eqref{eq:MFE} and the fixed points of the map $\Phi$: any fixed point of $\Phi$ provides a solution for equation \eqref{eq:MFE} and conversely any solution of equation \eqref{eq:MFE} constitute a fixed point of $\Phi$.

Let us start by proving that the possible solutions of the system have bounded second moment.

\begin{lemma}\label{lem:SoluL2}
	Let $\zeta_0 \in \M^2(C([-\tau,0],E^{P})$ a square-integrable process. Under the assumptions (H1)-(H4), there exists a constant $C(T)>0$ depending on the parameters of the system and on the horizon $T$, such that for any solution $X$ of the mean-field equation \eqref{eq:MFE} with initial condition $\zeta_0$:
	\[\sup_{[-\tau, T]}\Exp[ \vert X_t \vert^2]\leq C(T).\]
\end{lemma}
\begin{remark}
	The proof below corrects a misprint of the published version (useless stopping time $\tau_n$ appeared in the calculations). 
\end{remark}
\begin{proof}
	The proof is based on the application of It\^o's formula for the squared modulus of $X_t$, standard inequalities and Gronwall's lemma.

	Let $X$ be a solution of the mean-field equations. Due to the non-standard nature of the equation, let us underline the fact that It\^o's formula is valid, i.e. that $X_{t}$ is a semimartingale. By definition, it is clear that both $X_{t+s}$ and $Z_{t+s}$ are $\F_t$ measurable for all $s\in [-\tau,0]$, and hence necessarily $X_{t}$ is a semi-martingale, i.e. the sum of a continuous adapted process of finite variation:
	\[\int_0^t \Big(f_{\alpha}(s,X_s^{\alpha}) + \sum_{\gamma=1}^P \int_{-\tau}^0\Exp_{\bar{Z}}[b_{\alpha\gamma}(\bar{X}_s^{\alpha}, \bar{Z}^{\gamma}_{s+u})] \,d\eta_{\alpha\gamma}(u)\Big) \,ds\]
	and a continuous $(\F_t,\mathbbm{P})$-local martingale, defined as the sum of different stochastic integral of a progressively measurable processes with respect to Brownian motions:
	\[\int_0^t g_{\alpha}(s,\bar{X}^{\alpha}_s)\,dW_s^{\alpha} + \sum_{\gamma=1}^P \int_0^t \int_{-\tau}^0\Exp_{\bar{Z}}[\beta_{\alpha\gamma}(\bar{X}_s^{\alpha}, \bar{Z}^{\gamma}_{s+u})]\,d\eta_{\alpha\gamma}(u) \cdot dB_s^{\alpha\gamma}\]
	 We can hence apply It\^o formula to $\vert X_{t}\vert^2$, and because of the particular form of the squared norm, we can treat each component $\alpha\in\{1,\cdots,P\}$ separately. Since all the different Brownian motions involved are independent, we have for all $t>0$:
	\begin{multline*}
	\vert X_{t}^{\alpha}\vert^2=\vert \zeta_0^{\alpha}(0) \vert^2 + 2\int_0^{t} \Big\{ (X_s^{\alpha})^T f_{\alpha}(s,X_s^{\alpha}) + \frac 1 2 \vert g_{\alpha}(s,X_s^{\alpha})\vert^2  \\
	+ (X_s^{\alpha})^T \sum_{\gamma=1}^P \int_{-\tau}^0 \Exp_{Y}[b_{\alpha\gamma}(X_s^{\alpha},Y_{s+u}^{\gamma})] d\eta_{\alpha\gamma}(u)
	+
	\frac 1 2 \sum_{\gamma=1}^P  \vert \int_{-\tau}^0 \Exp_{Y}[\beta_{\alpha\gamma}(X_s^{\alpha},Y_{s+u}^{\gamma})] d\eta_{\alpha\gamma}(u) \vert^2 \Big\}\,ds \\
	+\int_0^{t}  (X_s^{\alpha})^T g_{\alpha}(s,X_s^{\alpha}) dW^{\alpha}_t +  \sum_{\gamma=1}^P \int_0^{t} (X_s^{\alpha})^T\int_{-\tau}^0 \Exp_{Y}[\beta_{\alpha\gamma}(X_s^{\alpha},Y_{s+u}^{\gamma})] d\eta_{\alpha\gamma}(u) dB^{\alpha\gamma}_s
	\end{multline*}

	The stochastic integral term has a null expectation, and the Stieljes integration term involves the term  $\vert x^T \,b_{\alpha\gamma}(x,z)\vert $, which is upperbounded because of assumption~\ref{Assump:bBound} by $\sqrt{\tilde{K}} (1+\vert x \vert^2)$, and the term $x^T \,f_{\alpha}(t,x)+\frac 1 2 \vert g_{\alpha}(t,x) \vert^2$ which is upperbounded, using assumption~\ref{Assump:MonotoneGrowth}, by $K(1+\vert x \vert^2)$. Finally, assumption~\ref{Assump:bBound} again allows us to upperbound the term $\frac 1 2 \Exp_{Y}[\vert\beta_{\alpha\gamma}(X_s^{\alpha},Y_s^{\gamma})\vert^2]$ by $\frac{\tilde{K}}{2} (1+\vert X_s^{\alpha} \vert^2)$.
	We hence have for all $t>0$:
	\begin{align*}
	\Exp[\vert X_{t}^{\alpha}\vert^2] &\leq \Exp[\vert \zeta_0^{\alpha}(0) \vert^2] + 2\int_0^{t} \Big\{ K(1+\Exp[\vert X_s^{\alpha}\vert^2])  + P \sqrt{\tilde{K}} (1+\Exp[\vert X_s^{\alpha}\vert^2]) \\
	&\qquad + \frac {\tilde{K}} 2  \sum_{\gamma=1}^P \int_{-\tau}^0  (1+\Exp[\vert X_s^{\alpha}\vert^2])  d\eta_{\alpha\gamma}(u)  \Big\}\,ds \\
	 &\leq \Exp[\vert \zeta_0^{\alpha}(0) \vert^2] + 2\int_0^{t} (K+  \sqrt{\tilde{K}} P + \frac {\tilde{K}} 2 P )(1+\Exp[\vert X_s^{\alpha}\vert^2])\,ds
	\end{align*}
	Summing over $\alpha$ and applying Gronwall's lemma directly yields:
	\[\sup_{t\in [0,{T}]} \Exp[\vert X_t\vert^2] \leq \Exp[\vert \zeta_0(0)\vert^2] +\Exp[1+\vert \zeta_0(0)\vert^2] (e^{K' T}-1)\]
	with $K'=2\,(K+\frac{\tilde{K}}{2} P + \sqrt{\tilde{K}} \, P)$, and hence conclude that:
	\[\sup_{t\in [-\tau,T]} \Exp[\vert X_t\vert^2] \leq \max(\sup_{[-\tau,0]} \Exp[\vert\zeta_0(s)\vert^2] \;\;,\;\; \Exp[\vert \zeta_0(0)\vert^2] + \Exp[1+\vert \zeta_0(0)\vert^2] (e^{K' T}-1))\]
\end{proof}

We now prove the existence and uniqueness theorem, using this boundedness property of the possible solutions.

\begin{theorem}\label{thm:ExistenceUniqueness}
	For any $\zeta_0 \in \M^2(C([-\tau,0],E^{P})$ a square-integrable process, the mean-field equation \eqref{eq:MFE} with initial condition $\zeta_0$ has a unique strong solution on $[-\tau,T]$ for any $T>0$.
\end{theorem}

\begin{proof}
	We start by showing the existence of solutions, and then prove the uniqueness property. We recall that by application of lemma~\ref{lem:SoluL2}, any possible solution has a bounded second moment. The proof is based on a classical contraction argument on iterates of the map $\Phi$. The main difference with the usual setting is the fact that we are considering delayed equations, which leads us to deal with uniform in $[-\tau,T]$ convergence properties. The rest of the proof uses classical inequalities such as Cauchy-Schwarz (CS), Burkholder-Davis-Gundy (BDG) and H\"older inequalities, all these applied on a decomposition of the term of interest into elementary terms.\\
	\noindent {\it Existence:}\\
	Let $X^0\in \M^2(\C)$ such that $X^0\vert_{[-\tau, 0]} \eqlaw \zeta_0$ a given stochastic process.

	We introduce the sequence of probability distributions $(X^k)_{k \geq 0}$ on $\M(\C)$ defined by induction as $X^{k+1}=(\Phi(X^k))$.
	We denote by $(Z^k)$ a sequence of processes independent of the collection of processes $(X^k)$ and having the same law.
	We start by controlling one of the components of the sequence of processes $(X^k_t)$, and for compactness of notations denote ${}^{\alpha}X^{k}_t \in E$ the component $\alpha$ of $X^{k}_t$. We decompose into six elementary terms the difference:
	\begin{align*}
		{}^{\alpha}X^{k+1}_t-{}^{\alpha}X^k_t &=\int_0^t \Big ( f_{\alpha}(s,{}^{\alpha}X^{k}_s) - f_{\alpha}(s,{}^{\alpha}X^{k-1}_s)\Big)\,ds \\
		& \quad\quad + \int_0^t \sum_{\gamma=1}^P  \int_{-\tau}^0 \Exp_Z \Big[ b_{\alpha\gamma} ({}^{\alpha}X^{k}_s, {}^{\gamma}Z^{k}_{s+u}) - b_{\alpha\gamma}({}^{\alpha}X^{k-1}_s, {}^{\gamma}Z^{k}_{s+u}) \Big]d\eta_{\alpha\gamma}(u) \, ds \\
		& \quad\quad + \int_0^t \sum_{\gamma=1}^P  \int_{-\tau}^0 \Exp_Z\Big [ b_{\alpha\gamma}({}^{\alpha}X^{k-1}_s, {}^{\gamma}Z^{k}_{s+u}) - b_{\alpha\gamma}({}^{\alpha}X^{k-1}_s, {}^{\gamma}Z^{k-1}_{s+u}) \Big] d\eta_{\alpha\gamma}(u) \, ds \\
		& \quad\quad +\int_0^t \Big ( g_{\alpha}(s,{}^{\alpha}X^{k}_s)-g_{\alpha}(s,{}^{\alpha}X^{k-1}_s) \Big) \, dW_s^{\alpha}\\
		& \quad\quad + \sum_{\gamma=1}^P \int_0^t \int_{-\tau}^0 \Exp_Z \Big [ \beta_{\alpha\gamma}({}^{\alpha}X^{k}_s,   {}^{\gamma}Z^{k}_{s+u})   - \beta_{\alpha\gamma}({}^{\alpha}X^{k-1}_s, {}^{\gamma}Z^{k}_{s+u}) \Big]d\eta_{\alpha\gamma}(u) \cdot dB_s^{\alpha\gamma} \\
		& \quad\quad + \sum_{\gamma=1}^P\int_0^t \int_{-\tau}^0 \Exp_Z  \Big [ \beta_{\alpha\gamma}({}^{\alpha}X^{k-1}_s, {}^{\gamma}Z^{k}_{s+u}) - \beta_{\alpha\gamma}({}^{\alpha}X^{k-1}_s, {}^{\gamma}Z^{k-1}_{s+u}) \Big]d\eta_{\alpha\gamma}(u) \cdot dB_s^{\alpha\gamma} \\
		& \stackrel{\text{def}}{=} A_t^{\alpha} + B_t^{\alpha} + C_t^{\alpha} + D_t^{\alpha} + E_t^{\alpha} + F_t^{\alpha}
	\end{align*}
	where we simply identify each of the six terms $A_t=(A_t^{\alpha}, \alpha=1,\cdots,P)$, $B_t$, $C_t$, $D_t$, $E_t$ and $F_t$ with the corresponding expression in the previous formulation. By a simple convexity inequality (H\"older) we have:
	\[\vert X^{k+1}_t-X^k_t \vert^2 \leq 6 \Big( \vert A_t\vert^2 + \vert B_t\vert^2 + \vert C_t\vert^2+ \vert D_t\vert^2+\vert E_t\vert^2+\vert F_t\vert^2\Big)\]
	and treat each term separately.

	The term $A_t$ is easily controlled using Cauchy-Schwarz inequality, Fubini identity and standard inequalities and we obtain:
	\begin{align*}
		\Exp \Big[\sup_{\sup_{s\in[0,t]}} \vert A_s \vert^2\Big]
		&=	\Exp \Big[\sup_{s\leq t} \sum_{\alpha=1}^P \vert \int_0^s f_{\alpha}(u,{}^{\alpha}X^{k}_u) - f_{\alpha}(u,{}^{\alpha}X^{k-1}_u) \, du \vert^2 \Big ]\\
						&{\leq \Exp \Big[\sup_{s\leq t}  s\, \sum_{\alpha=1}^P \int_0^s \vert f_{\alpha}(s,{}^{\alpha}X^{k}_u) - f_{\alpha}(s,{}^{\alpha}X^{k-1}_u)\vert^2 \, ds  \Big ]}\\
						& {\leq K^2 \,t \, \Exp \Big[\sup_{s\leq t} \sum_{\alpha=1}^P\int_0^s \vert {}^{\alpha}X^{k}_u - {}^{\alpha}X^{k-1}_u \vert^2 \, ds \Big ]}\\
						& {\leq K^2 \,t \, \Exp \Big[ \int_0^t \vert X^{k}_s - X^{k-1}_s \vert^2 \, ds \Big ]}\\
		& \leq K^2 \,t \, \int_0^t \Exp \Big[ \sup_{-\tau\leq u\leq s} \vert X^{k}_u - X^{k-1}_u \vert^2\Big ] \, ds
	\end{align*}
	Similarly, the martingale term $D_t$ is bounded using the Burkholder-Davis-Gundy theorem to the $Pd$-dimensional martingale $(\int_0^t ( g_{\alpha}(s,{}^{\alpha}X^{k}_s)-g_{\alpha}(s,{}^{\alpha}X^{k-1}_s)) \, dW_s^{\alpha} , \alpha=1\ldots P)$ and we obtain:
	\begin{align*}
		\Exp \Big[\sup_{0\leq s\leq t} \vert D_s \vert^2\Big] &= {\Exp \Big[\sup_{s\leq t} \sum_{\alpha=1}^P\vert \int_0^s g_{\alpha}(u,{}^{\alpha}X^{k}_u) - g_{\alpha}(u,{}^{\alpha}X^{k-1}_u) \, dW^{\alpha}_u \vert^2 \Big]}\\
		 &{\leq 4 \int_0^t \Exp \Big[  \sum_{\alpha=1}^P \vert g_{\alpha}(s,{}^{\alpha}X^{k}_s) -g_{\alpha}(s,{}^{\alpha}X^{k-1}_s)\vert^2 \Big]  \, ds  }\\
		 &{\leq 4\,K^2 \, \int_0^t \Exp \Big[ \vert X^{k}_s - X^{k-1}_s \vert^2 \Big]  \, ds } \\
		& {\leq 4 K^2 \, \int_0^t \Exp \Big[ \sup_{-\tau \leq u\leq s} \vert X^{k}_u - X^{k-1}_u \vert^2 \Big]  \, ds }
	\end{align*}

	Let us now deal with the deterministic interaction terms $B_t$ and $C_t$. We have:
	\begin{align*}
		\vert B_t \vert^2 &=\sum_{\alpha=1}^P \left \vert \int_0^t \sum_{\gamma=1}^P \int_{-\tau}^0 (\Exp_Z[b_{\alpha\gamma}({}^{\alpha}X^{k}_s,{}^{\gamma}Z^{k}_{s+u}) - b_{\alpha\gamma}({}^{\alpha}X^{k-1}_s,{}^{\gamma}Z^{k}_{s+u})])\,d\eta_{\alpha\gamma}(u)\,ds \right\vert^2\\
		& \leq \sum_{\alpha=1}^P t\,P\,  \int_0^t \sum_{\gamma=1}^P \int_{-\tau}^0 \Exp_Z\left[ \vert b_{\alpha\gamma}({}^{\alpha}X^{k}_s,{}^{\gamma}Z^{k}_{s+u}) - b_{\alpha\gamma}({}^{\alpha}X^{k-1}_s,{}^{\gamma}Z^{k}_{s+u})\vert^2\right])
		\,d\eta_{\alpha\gamma}(u)\,ds\\
		&\leq t P^2  L^2 \int_0^t \vert X^k_s - X^{k-1}_s\vert^2\,ds \leq t P^2  L^2 \int_0^t \sup_{-\tau \leq u\leq s}\vert X^k_u - X^{k-1}_u\vert^2\,ds
	\end{align*}
	hence easily conclude that
	\[\Exp[\sup_{s\in[0,t]}\vert B_s \vert^2] \leq t P^2  L^2 \int_0^t \Exp[\sup_{-\tau \leq u\leq s}\vert X^k_u - X^{k-1}_u\vert^2]\,ds \]

	and similarly, since $Z^k$ and $Z^{k-1}$ have the same law as $X^{k}$ and $X^{k-1}$,
	\[\Exp[\sup_{s\in[0,t]}\vert C_s \vert^2] \leq t P^2  L^2 \int_0^t \Exp[\sup_{-\tau \leq u\leq s}\vert X^k_u - X^{k-1}_u\vert^2]\,ds \]

	Eventually, the terms $E_t$ and $F_t$ use Burkholder-David-Gundy (BDG) inequality in place of Cauchy-Schwarz' together with similar arguments as used for $B_t$ and $C_t$. Using the independence of the Brownian motions $(B^{\alpha\gamma}_t)$, BDG inequality yields, for the term $E_t$ (and similarly for the term $F_t$):
	\[\Exp[\sup_{s\in[0,t]}\vert \Theta_s \vert^2] \leq 4 P^2  L^2 \int_0^t \Exp[\sup_{-\tau \leq u\leq s}\vert X^k_u - X^{k-1}_u\vert^2]\,ds \]
	for $\Theta_t$ equal to $E_t$ or $F_t$.
	Putting all these evaluations together, we get:
	\begin{equation}\label{eq:Bounds}
		\Exp\Big[\sup_{s\in[0,t]} \vert X^{k+1}_s-X^k_s \vert^2 \Big] \leq 6(T+4)(K^2 + 2P^2\,L^2\, ) \int_0^t \Exp[ \sup_{-\tau \leq u\leq s} \vert X^{k}_u - X^{k-1}_u \vert^2] ds
	\end{equation}
Moreover, since $ X^{k+1}_t \equiv X^{k}_t$ for $t\in [-\tau, 0]$ by definition, we have, noting
\[M^k_t =\Exp\Big[\sup_{-\tau \leq s\leq t} \vert X^{k+1}_s-X^k_s \vert^2 \Big],\]
the recursive inequality $M^k_t \leq K'' \int_0^t M^{k-1}_s\,ds$ with $K''=6(T+4)(K^2 + 2P^2 L^2 )$. We thus get by an immediate recursion:
\begin{align}
	\nonumber M^k_t &\leq (K'')^k \int_0^t\int_0^{s_1}\cdots \int_0^{s_{k-1}}M^0_{s_k}\;ds_1\ldots ds_k\\
	\label{eq:CauchySeq} & \leq \frac{(K'')^k\, t^k}{k!} M^0_T
\end{align}
and $M^0_t$ is finite because of lemma~\ref{lem:SoluL2}. Routine methods starting from inequality~\eqref{eq:CauchySeq} and allows proving existence and uniqueness of fixed point for $\Phi$ (see~\cite[pp. 376--377]{revuz-yor:99}).
Indeed, the fact that
\[\sum_{n=1}^{\infty} \Ex{\sup_{s\in[-\tau,T]}\vert X^{k+1}_s-X^k_s \vert^2 }^{1/2} <\infty,\]
implies the almost sure convergence of the series $\sum_{n=0}^{\infty} \sup_{s\in[-\tau,T]}\vert X^{k+1}_s-X^k_s \vert$ and hence the a.s. uniform convergence of the partial sums:
\[X^0_t+\sum_{k=0}^n (X^{k+1}_t-X^k_t) = X^n_t\]
on $[-\tau,T]$. Denote by $\bar{X}_t$ the thus defined limit. It is clearly almost surely continuous and $\F_t$ adapted. From lemma \ref{lem:SoluL2}, we know that $\bar{X} \in \M^2(\C)$. It remains to show that $\bar{X}$ indeed satisfies the equation \eqref{eq:MFE}. This is easily done using the estimates we derived. Indeed, it is easy to show at this point that
\[\Exp[\sup_{-\tau \leq t\leq T}\vert \Phi(X^n)_t - \Phi(\bar{X})_t\vert^2] \leq K''\int_0^T \Exp[\sup_{-\tau \leq t\leq T}\vert X^n_s-\bar{X}_s\vert^2]\,ds \rightarrow 0\]
as $n \to \infty$. Hence $\Phi(\bar{X})$ is equal to the limit of the sequence $\Phi(X^n)$, which is by definition, equal to $X^{n+1}$ whose limit is $\bar{X}$. This concludes the fact that $\bar{X}$ is a fixed point of $\Phi$ which completes the proof of the existence of solutions of the mean-field equation \eqref{eq:MFE}.

\noindent{\it Uniqueness:}\\
Assume that $X$ and $Y$ are two solutions of the mean-field equations \eqref{eq:MFE}. From lemma \ref{lem:SoluL2}, we know that both solutions are in $\M^2(\C)$.
Using the bound \eqref{eq:Bounds}, we directly obtain the inequality:
\[\Exp\Big[\sup_{-\tau\leq s\leq T} \vert X_s-Y_s \vert^2 \Big] \leq K'' \int_0^{T} \Exp \Big[ \sup_{-\tau\leq u\leq s} \vert X_u - Y_u \vert^2 \Big]  \, ds\]
which by Gronwall's theorem directly implies $\Exp\Big[\sup_{-\tau\leq s\leq T} \vert X_s-Y_s \vert^2 \Big]=0$ whence $X=Y$ on $[-\tau,T]$ a.s. follows.
\end{proof}

\section{Convergence results}\label{sec:limitlaw}
We now turn to the main result, namely the convergence in law of the solutions of the network equations towards the equations~\eqref{eq:MFE}. More precisely, we are interested in showing that the propagation of chaos applies in this case. The propagation of chaos property states that provided that all neurons have a chaotic initial condition (i.e., i.i.d. in $\M^2(\Ct)$), then in the limit where the number of neurons is infinite, all neurons behave independently and the limit process satisfies the equation given by \eqref{eq:MFE}. The proof follows standard proofs in the domain as generally done in particular by Tanaka or Sznitman \cite{tanaka:78,sznitman:84a}. It is based on the very powerful coupling argument, that identifies the almost sure limit of the network solutions $X^{i,N}$ (the exponent $N$ denotes the dependence of the law of neuron $i$ upon the network size $N$) as the number of neuron tends to infinity, a method popularized by Sznitman in his extensive works (see e.g.~\cite{sznitman:89}), the ideas of which traces back to the 70's (for instance, Dobrushin uses it in~\cite{dobrushin:70}). We start by considering the translation invariant case where the delay measures $\eta_{i\beta}$ only depend on the neural population of neuron $i$, prior to dealing with the non translation invariant case.

\subsection{Quenched propagation of chaos for translation invariant neural fields}
We consider here $\eta_{i\gamma}=\eta_{\alpha\gamma}$ for all $i$ in population $\alpha$. We intend to prove the propagation of chaos property and the convergence towards the mean field equations for almost all realization of the the delays $\tau_{ij}$ under the assumption that for any fixed $i$, the random variables $(\tau_{ij})_{j\in\{1\cdots N\}}$ are independent. This is termed the quenched propagation of chaos property. In order to demonstrate this property, we define a coupling between the solutions of the network equations~\eqref{eq:Network} and the mean-field equations~\eqref{eq:MFE}. Let $i\in\N$ such that $p(i)=\alpha$. We define the process $\bar{X}^i$ solution of the mean-field equation~\eqref{eq:MFE}, driven by the Brownian motions $(W^i_t)$ and $(B^{i\gamma}_t)$ that govern the network process $X^i$, and having the same initial condition as neuron $i$ in the network, $\zeta^i_0 \in \M^2(\Ct)$. In details, $\bar{X}^i_t$ is the unique solution of the equation, for $\alpha=p(i)$:
\begin{equation}\label{eq:Coupling}
	\left \{
	\begin{array}{lll}
			d\bar{X}^i_t &= f_{\alpha}(t,\bar{X}^i_t)\,dt + \sum_{\gamma=1}^P\int_{-\tau}^0\Exp_Z[b_{\alpha\gamma}(\bar{X}^i_t,Z_{t+s}^{\gamma})]d\eta_{\alpha\gamma}(s)\,dt \\
			& \quad + g_{\alpha}(t,\bar{X}^{i}_t)\, dW^i_t
			+ \sum_{\gamma=1}^P\int_{-\tau}^0\Exp_Z[\beta_{\alpha\gamma}(\bar{X}^i_t,Z_{t+s}^{\gamma})]\,d\eta_{\alpha\gamma}(s)\,dB^{i\gamma}_t  & t\geq 0\\
		\bar{X}^i_t & = \zeta^i_0 (t) & t\in [-\tau, 0]
	\end{array}
	\right .
	,
\end{equation}
which constitutes a collection of independent stochastic processes $(\bar{X}^i_t)_{i=1\ldots N}$ that are identically distributed with law $\bar{X}^{p(i)}$.
Let us denote by $m_t^{\alpha}$ the probability distribution of $\bar{X}_t^{\alpha}$ solution of the mean-field equation \eqref{eq:MFE}. As previously, the processes $(Z^1_t,\cdots,Z^P_t)$ constitute a collection of processes independent of $(\bar{X}^i_t)_{i=1\ldots N}$ and having the distribution $m^1\otimes \cdots \otimes m^P$. We aim at proving the almost sure convergence of a collection of processes  $(\mathbbm{X}^N_t):=({X}^{i_k,N}_t, k=1\ldots l)$ for some $l\in \N^*$
and $(i_1,\cdots,i_l)\in\N^l$ towards $(\bar{X}^{i_k}_t)$, implying the convergence of the law of $(\mathbbm{X}^N_t)$ towards $m^{p(i_1)} \otimes\ldots\otimes m^{p(i_k)}$ as $N$ goes to infinity. We start by proving this result for $l=1$ which will readily imply the case $l\geq 1$. In what follows, we shall denote by $\E$ the expectation over the delays $\tau_{ij}$.

\begin{theorem}[Convergence in law]\label{thm:PropagationChaos}
	Under the assumptions (H1)-(H4) and the chaotic initial condition assumption, the process $(X^{i,N}_t, -\tau \leq t\leq T)$, solution of the network equations~\eqref{eq:Network}, converge almost surely towards the process $(\bar{X}^i_t, -\tau \leq t\leq T)$ solution of the mean-field equations~\eqref{eq:Coupling}. This implies in particular convergence in law of the process $(X^{i,N}_t, -\tau \leq t\leq T)$ towards $(\bar{X}^{\alpha}_t, -\tau \leq t \leq T) $
	solution of the mean-field equations~\eqref{eq:MFE}. Moreover, since $f$ and $g$ are globally Lipschitz-continuous, we have for any $i\in \N$ and any $T>0$:
\begin{equation}\label{eq:Propchaos}
	\max_{i=1\cdots N} \E\left\{\Exp\Big [\sup_{-\tau \leq s\leq T} \vert X^{i,N}_s - \bar{X}^i_s\vert^2 \Big]\right\}< \frac{C(T)}{\min_{\gamma}(N_{\gamma})}
\end{equation}
where $C(T)$ is a constant only depending on the parameters of the system and the time horizon $T$.
\end{theorem}

\begin{proof}
	The proof differs from more classical proofs in that we consider that the network is composed of several distinct populations and includes random delays. The principle of the proof consists in thoroughly analyzing the difference between the two processes as $N$ tends to infinity. This difference is written as the sum of eight terms denoted $A_t(N)$ through $H_t(N)$:
	\begin{align}
		\nonumber X^{i,N}_t-\bar{X}^i_t&=\int_0^t (f_{\alpha}(s,X^{i,N}_s)-f_{\alpha}(s,\bar{X}^i_s)) \, ds + \int_0^t (g_{\alpha}(s,X^{i,N}_s)-g_{\alpha}(s,\bar{X}^i_s)) \, dW^i_s\\
		\nonumber &\qquad + \sum_{\gamma=1}^P \int_0^t \Big(\frac 1 {N_{\gamma}} \sum_{p(j)=\gamma} b_{\alpha\gamma}(X^{i,N}_s,X^{j,N}_{s-\tau_{ij}})-b_{\alpha\gamma}(\bar{X}^i_s,X^{j,N}_{s-\tau_{ij}})\Big) \,ds\\
		\nonumber &\qquad + \sum_{\gamma=1}^P \int_0^t \Big( \frac 1 {N_{\gamma}} \sum_{p(j)=\gamma} b_{\alpha\gamma}(\bar{X}^i_s,X^{j,N}_{s-\tau_{ij}})-b_{\alpha\gamma}(\bar{X}^i_s,\bar{X}^j_{s-\tau_{ij}})\Big) \, ds\\
		\nonumber &\qquad + \sum_{\gamma=1}^P\int_0^t \Big(\frac 1 {N_{\gamma}} \sum_{p(j)=\gamma}  b_{\alpha\gamma}(\bar{X}^i_s,\bar{X}^j_{s-\tau_{ij}})-\int_{-\tau}^0\Exp_Z[b_{\alpha\gamma}(\bar{X}^i_s,Z^{\gamma}_{s+u})]\,d\eta_{\alpha\gamma}(u)\Big)\, ds\\
		\nonumber &\qquad + \sum_{\gamma=1}^P\int_0^t \Big(\frac 1 {N_{\gamma}} \sum_{p(j)=\gamma} (\beta_{\alpha\gamma}(X^{i,N}_s,X^{j,N}_{s-\tau_{ij}})-\beta_{\alpha\gamma}(\bar{X}^i_s,X^{j,N}_{s-\tau_{ij}}))\, \Big)dB^{i\gamma}_s\\
		\nonumber &\qquad + \sum_{\gamma=1}^P\int_0^t  \Big( \frac 1 {N_{\gamma}} \sum_{p(j)=\gamma} \beta_{\alpha\gamma}(\bar{X}^i_s,X^{j,N}_{s-\tau_{ij}})-\beta_{\alpha\gamma}(\bar{X}^i_s,\bar{X}^j_{s-\tau_{ij}})\Big)\, dB^{i\gamma}_s
		\end{align}
	\begin{align}
		\nonumber &\qquad + \sum_{\gamma=1}^P\int_0^t  \Big(\frac 1 {N_{\gamma}} \sum_{p(j)=\gamma}\beta_{\alpha\gamma}(\bar{X}^i_s,\bar{X}^j_{s-\tau_{ij}})-\int_{-\tau}^0 \Exp_Z[\beta_{\alpha\gamma}(\bar{X}^i_s,Z^{\gamma}_{s+u})]\,d\eta_{\alpha\gamma}(u)\Big)\, dB^{i\gamma}_s\\
		\label{eq:EightTerms}&= A_t(N)+B_t(N)+C_t(N)+D_t(N)+E_t(N)+F_t(N)+G_t(N)+H_t(N)
	\end{align}

	Let us underline the fact that the probability distribution of these terms do not depend on the specific neuron $i$ in population $\alpha$ considered. We are interested in the limit, as $N$ goes to infinity, of the quantity $\E\{\Exp [\sup_{-\tau \leq s\leq T} \vert X^{i,N}_s - \bar{X}^i_s\vert^2 ]\}$. We decompose this expression into the sum of the eight terms involved in equation \eqref{eq:EightTerms} using H\"older's inequality and upperbound each term separately. The terms $A_t(N)$ and $B_t(N)$ are treated exactly as in the proof of theorem \ref{thm:ExistenceUniqueness}, and the control all the terms except $E_t$ and $H_t$ essentially uses on the same ingredients as in the proof of theorem~\ref{thm:ExistenceUniqueness}. Similarly to the analysis performed in the proof of theorem~\ref{thm:ExistenceUniqueness}, we have the following inequalities:
	\begin{equation*}
		\begin{cases}
			\Exp[\sup_{0\leq s\leq t\wedge \tau_U} \vert A_s(N) \vert^2] & \leq K^2\, T\, \int_0^{t\wedge \tau_U} \Exp[\sup_{-\tau\leq u\leq s} \vert X_u^{i,N}-\bar{X}_u^i\vert^2 ]\, du\\
			\Exp[\sup_{0\leq s\leq t\wedge \tau_U} \vert B_s(N) \vert^2] & \leq 4\, K^2\, \int_0^{t\wedge \tau_U} \Exp[\sup_{-\tau\leq u\leq s} \vert X_u^{i,N}-\bar{X}_u^i\vert^2 ]\, du\\
			\Exp[\sup_{0\leq s\leq t\wedge \tau_U} \vert C_s(N) \vert ^2] & \leq T L^2 P^2  \int_0^{t\wedge \tau_U}  \Exp[\sup_{-\tau\leq u\leq s}\vert X^{i,N}_u-\bar{X}^i_u \vert^2] \, ds\\
			\Exp[\sup_{0\leq s\leq t\wedge \tau_U} \vert D_s(N) \vert ^2] & \leq T\,L^2 P^2  \int_0^{t\wedge \tau_U} \max_{j=1\cdots N} \Exp[\sup_{-\tau \leq u \leq s}\vert X^{j,N}_u-\bar{X}^j_u \vert^2] \, ds\\
			\Exp[\sup_{0\leq s\leq t\wedge \tau_U}\vert F_s(N) \vert ^2] &\leq 4 L^2 P^2  \int_0^t  \Exp[\sup_{-\tau\leq u\leq s}\vert X^{i,N}_u-\bar{X}^i_u \vert^2] \, ds\\
			\Exp[\sup_{0\leq s\leq t\wedge \tau_U}\vert G_s(N) \vert ^2] &\leq 4 L^2 P^2  \int_0^t \max_{j=1\cdots N} \Exp[\sup_{-\tau\leq  u \leq s}\vert X^{j,N}_u-\bar{X}^j_u \vert^2] \, ds
		\end{cases}
	\end{equation*}
	Let us for instance treat the case of $F_t (N)$ for the sake of clarity, all other terms are treated along the same lines. We have:
	\begin{align*}
	\Exp[\sup_{0 \leq s\leq t}\vert F_s(N) \vert ^2] &\leq P \sum_{\gamma=1}^P \Exp\bigg[\sup_{0 \leq s\leq t} \Big\vert \int_0^s \frac 1 {N_{\gamma}} \sum_{p(j)=\gamma} \beta_{\alpha\gamma}(X^{i,N}_u,X^{j,N}_{u-\tau_{ij}})\\
	& \quad -\beta_{\alpha\gamma}(\bar{X}^{i}_u,X^{j,N}_{u-\tau_{ij}}) dB^{i\gamma}_u\Big\vert^2\bigg]\\
	\text{(BDG)}&\leq 4P \sum_{\gamma=1}^P \Exp\left[\int_0^t \left\vert\frac 1 {N_{\gamma}} \sum_{p(j)=\gamma} \beta_{\alpha\gamma}(X^{i,N}_s,X^{j,N}_{s-\tau_{ij}})-\beta_{\alpha\gamma}(\bar{X}^{i}_s,X^{j,N}_{s-\tau_{ij}}) \right\vert^2d_s\right]\\
	\text{(Cauchy-Schwartz)}&\leq 4P \sum_{\gamma=1}^P \int_0^t \frac 1 {N_{\gamma}}  \sum_{p(j)=\gamma} \Exp\left[\left\vert \beta_{\alpha\gamma}(X^{i,N}_s,X^{j,N}_{s-\tau_{ij}})-\beta_{\alpha\gamma}(\bar{X}^{i}_s,X^{j,N}_{s-\tau_{ij}})\,\right\vert^2\right]\,ds\\
	\text{(assumption~\ref{Assump:LocLipschb})} &\leq 4P^2 L^2 \int_0^t \Exp\left[\left\vert X^{i,N}_s-\bar{X}^i_s\right\vert^2\right]\,ds\\
	&\leq 4 P^2 L^2  \int_0^t \Exp\left[\sup_{-\tau \leq u\leq s}\vert X^{i,N}_u-\bar{X}^i_u \vert^2\right] \, ds.
	\end{align*}

	In the decomposition~\eqref{eq:EightTerms}, the two terms $E_t(N)$ and $H_t(N)$ are of a new form and remain to be evaluated. Both term involves the difference between an empirical mean of a function of processes and an expectation term, and all have bounded second moment thanks to proposition~\ref{lem:SoluL2} and assumption~\ref{Assump:bBound}. We have:
	\begin{align}
		\nonumber \E \{ \Exp[ \sup_{0\leq s\leq t} \vert E_s(N) \vert ^2]\} &=\E\bigg\{\Exp \Big[ \sup_{0\leq s\leq t} \Big\vert \int_0^s \sum_{\gamma=1}^P \Big( \frac 1 {N_{\gamma}} \sum_{p(j)=\gamma} b_{\alpha\gamma}(\bar{X}^i_s,\bar{X}^j_{s-\tau_{ij}}) -\int_{-\tau}^0\Exp_Z[b_{\alpha\gamma}(\bar{X}^i_s,Z_{s+u}^{\gamma})]\,d\eta_{\alpha\gamma}(u) \Big) \, ds\Big\vert^2\Big]\bigg\}\\
		\label{eq:EtHtControl}& \leq T P  \sum_{\gamma=1}^P \int_0^t \E\Big\{ \Exp\bigg[\bigg \vert \frac 1 {N_{\gamma}} \sum_{p(j)=\gamma} b_{\alpha\gamma}(\bar{X}^i_s,\bar{X}^j_{s-\tau_{ij}})- \int_{-\tau}^0\Exp_Z[b_{\alpha\gamma}(\bar{X}^i_s,Z_{s+u}^{\gamma})\, d\eta_{\alpha\gamma}(u)] \bigg\vert^2\bigg] \, ds\Big\}
	\end{align}
	and using Burkholder-Davis-Gundy
	\[
		\E\Big\{\Exp[\sup_{s\leq t}\vert H_s(N) \vert ^2]\Big\} \leq 4 P  \sum_{\gamma=1}^P \int_0^t  \E\Big\{\Exp[\bigg \vert \frac 1 {N_{\gamma}} \sum_{p(j)=\gamma} \beta_{\alpha\gamma}(\bar{X}^i_s,\bar{X}^j_{s-\tau_{ij}})- \int_{-\tau}^0\Exp_Z[\beta_{\alpha\gamma}(\bar{X}^i_s,Z_{s+u}^{\gamma}) \, d\eta_{\alpha\gamma}(u)] \bigg\vert^2]\Big\}\, ds
		\]
		Each of these expressions involve the processes $X^j_{t-\tau_{ij}}$ which, under the tensor product law of the Brownian motion and the delays, are independent and identically distributed population-wise. Moreover, we have:
	\begin{multline*}
		\E\Big\{\Exp[\vert \frac 1 {N_{\gamma}} \sum_{p(j)=\gamma} \Theta(\bar{X}^i_s,\bar{X}^j_{s-\tau_{ij}})-\int_{-\tau}^0 \Exp_Z[\Theta(\bar{X}^i_s,Z_{s+u}^{\gamma})]d\eta_{\alpha\gamma}(u)]\vert^2]\Big\} \\= \frac 1 {N_{\gamma}^2} \sum_{k,l =1}^{N_{\gamma}} \E\Big\{\Exp\Big[(\Theta(\bar{X}^i_s,\bar{X}^j_{s-\tau_{ij}})-\Exp_{Z,\tilde{\tau}_{\alpha\gamma}}[\Theta(\bar{X}^i_s,Z_{s-\tilde{\tau}_{\alpha\gamma}}^{\gamma})])^T \cdot (\Theta(\bar{X}^i_s,\bar{X}^k_{s-\tau_{ij}})-\Exp_{Z,\tilde{\tau}}[\Theta(\bar{X}^i_s,Z_{s-\tilde{\tau}_{\alpha\gamma}}^{\gamma})])\Big]\Big\}
	\end{multline*}
	where $\Theta\in\{b_{\alpha\gamma}, \beta_{\alpha\gamma}\} $
	, noting that the expectation $\E\{\Exp_{Z}[\cdot]\}$ of the composed random variable $\Theta(\bar{X}^i_s,\bar{X}^j_{s-\tau_{ij}})$ under the law of $\bar{X}^j$ and of the delays $\tilde{\tau}=\{\tau_{ij}\}$ is precisely equal to $\int_{-\tau}^0 \Exp_Z[\Theta(\bar{X}^i_s,Z_{s+u}^{\gamma})]d\eta_{\alpha\gamma}(u)]$.
	 In the above expression, $\tilde{\tau}_{\alpha\gamma}$ denotes a random variable with law $\eta_{\alpha\gamma}$ independent of the sequence of delays and of the Brownian motions. It is now easy to show that all the terms of the sum corresponding to indexes $j$ and $k$ such that the three conditions $j\neq i$, $k\neq i$ and $j\neq k$ are satisfied are null. The simpler way to see this property is to write the expectations as integrals with respect to the measure $m^{\alpha}_t$ and observe that all terms annihilate, using the mutual independence of the processes $(\bar{X}^j)$, of the random variables $(\tau_{ij})_{j\in\{1\cdots N\}}$,
	 and the independence between the processes and the delays.
	Therefore, there are no more than $3\,N_{\gamma}-1$ non-null terms in the sum (in the case $\alpha=\gamma$ there are just $N_{\gamma}$ non-null terms), and moreover, all of these terms are uniformly bounded. The terms related to indexes $j =k \neq i $ satisfy the inequality:
	\begin{align*}
		\E\left\{\Exp\left[\left \vert\Theta(\bar{X}^i_s,\bar{X}^j_{s-\tau_{ij}})-\Exp_{Z,\tilde{\tau}_{\alpha\gamma}}[\Theta(\bar{X}^i_s,Z_{s-\tilde{\tau}_{\alpha\gamma}}^{\gamma})]\right \vert^2\right]\right\} &\leq 2\,\E\left\{\Exp\left[\left \vert\Theta(\bar{X}^i_s,\bar{X}^j_{s-\tau_{ij}})\right \vert^2 + \left \vert\Exp_{Z,\tilde{\tau}_{\alpha\gamma}}[\Theta(\bar{X}^i_s,Z_{s-\tilde{\tau}_{\alpha\gamma}}^{\gamma})]\right \vert^2\right]\right\}\\
		&\leq 4 \tilde{K} (1+C_2(s))
	\end{align*}
	where $C_2(s)$ is the upperbound of the $\mathbbm{L}^2$ norm of $\bar{X}^{i}$ given by lemma~\ref{lem:SoluL2}. The terms related to the cases $j=i$ (or symmetrically $k=i$) are bounded by the same constant, since we have for all $k$ such that $p(k)=\alpha$, by Cauchy-Schwartz inequality:
	\begin{align*}
		& \E\left\{\Exp\left[\left (\Theta(\bar{X}^i_s,\bar{X}^i_{s-\tau_{ij}})-\Exp_{Z,\tilde{\tau}_{\alpha\gamma}}[\Theta(\bar{X}^i_s,Z_{s-\tilde{\tau}_{\alpha\gamma}}^{\alpha})]\right )^T \cdot \left (\Theta(\bar{X}^i_s,\bar{X}^k_{s-\tau_{ij}})-\Exp_{Z,\tilde{\tau}_{\alpha\gamma}}[\Theta(\bar{X}^i_s,Z_{s-\tilde{\tau}_{\alpha\gamma}}^{\alpha})]\right )^T \right] \right\}\\
		& \qquad \leq \E\left\{\Exp\left[\left \vert\Theta(\bar{X}^i_s,\bar{X}^i_{s-\tau_{ij}})-\Exp_{Z,\tau}[\Theta(\bar{X}^i_s,Z_{s-\tau}^{\alpha})]\right \vert^2\right]\right\}^{1/2} \E\left\{\Exp\left[\left \vert\Theta(\bar{X}^i_s,\bar{X}^k_{s-\tau_{ij}})-\Exp_{Z,\tau}[\Theta(\bar{X}^i_s,Z_{s-\tau}^{\alpha})]\right \vert^2\right]\right\}^{1/2}\\
		&\qquad \leq 4 \tilde{K} (1+C_2(s))
	\end{align*}

	We note $C= 4 \tilde{K} (1+C_2(T))$. We have:
	\[\E\left\{\Exp[\sup_{0\leq s\leq t}\vert E_s(N) \vert ^2]\right\} \leq T^2P C \sum_{\gamma=1}^P\frac{3N_{\gamma}-2}{N_{\gamma}^2} \leq 3\,T^2\,P^2\, C \frac{1}{\min_{\gamma}(N_{\gamma})}  \]
	and
	\[\E\left\{\Exp[\sup_{0\leq s\leq t}\vert H_s(N) \vert ^2]\right\} \leq 12\,T\,P^2\, C \frac{1}{\min_{\gamma}(N_{\gamma})}  \]

	Assembling all the estimates, we obtain
	\begin{align*}
		\max_{i=1\cdots N}\E\{\Exp[\sup_{0 \leq s\leq t} \vert X^{i,N}_s-\bar{X}^i_s\vert^2]\} & \leq 8\, \bigg\{(K^2+2 L^2 P^2 )\,(T+4) \int_0^t\max_{j=1\cdots N}\Exp\Big[\sup_{-\tau \leq u\leq s}\vert X^{j,N}_u - \bar{X}^j_u\vert^2\Big]\, du \\
		&\qquad + 3(T+4)T\frac{C\,P^2}{\min_{\gamma}(N_{\gamma})}\Big\}
	\end{align*}
	Let us denote $K'=8(K^2+2 L^2 P^2 )\,(T+4)$ and $C'=24 (T+4) C\,P^2\,$. Since $X^{i,N}_t$ and $\bar{X}^i_t$ are equal on $[-\tau,0]$, this inequality implies:
	\[\max_{i=1\cdots N}\Exp[\sup_{-\tau \leq s\leq t} \vert X^{i,N}_s-\bar{X}^i_s\vert^2] \leq K' \int_0^t\max_{j=1\cdots N}\Exp\Big[\sup_{-\tau \leq u\leq s}\vert X^{j,N}_u - \bar{X}^j_u\vert^2\Big]\, du + \frac{C'}{\min_{\gamma}(N_{\gamma})}\]
	Using Gronwall's inequality, we obtain:
	\begin{equation}\label{eq:Gron1}
		{\max_{i=1\cdots N}\Exp[\sup_{-\tau \leq s\leq t} \vert X^{i,N}_s-\bar{X}^i_s\vert^2] \leq \frac{C'}{\min_{\gamma}(N_{\gamma})} e^{K' T} }
	\end{equation}
	which tends to zeros as $N$ goes to infinity by assumption. This property implies the convergence in law of $(X^{i,N}_{t\wedge \tau_U}, -\tau\leq t\leq T)$ towards $(\bar{X}^{i}_{t}, -\tau\leq t\leq T)$.
%
	%
\end{proof}

\begin{corollary}[Propagation of chaos]\label{cor:PropaChaos}
	Let $l\in \N^*$ and fix $l$ neurons $(i_1,\cdots,i_l) \in \N^*$. Under the assumptions of theorem~\ref{thm:PropagationChaos}, the law of $(X^{i_1,N}_t, \cdots, X^{i_l,N}_t, -\tau\leq t \leq T)$ converges towards $m^{p(i_1)}\otimes \cdots \otimes m^{p(i_l)}$.
\end{corollary}

\begin{proof}
	We have, using the notations of the proof of theorem~\ref{thm:PropagationChaos},
	\begin{align*}
		& \E\left\{\Exp \left[ \sup_{-\tau\leq t \leq T} \left\vert (X^{i_1,N}_t, \cdots, X^{i_l,N}_t) - (\bar{X}^{i_1}_t, \cdots, \bar{X}^{i_l}_t)\right\vert^2 \right]\right\} \\
		& \qquad \leq \sum_{k=1}^l \E\left\{\Exp \left[ \sup_{-\tau\leq t \leq T\wedge\tau_U} \left\vert X^{i_k,N}_t-\bar{X}^{i_k}_t\right\vert^2 \right]\right\}\\
		& \qquad \leq \frac{l\,C'}{K' {\min_{\gamma}(N_{\gamma})}} e^{K' T}
	\end{align*}
	which tends to zero as $N$ goes to infinity, hence the law of $(X^{i_1,N}_t, \cdots, X^{i_l,N}_t, -\tau\leq t \leq T)$  converges towards that of $(\bar{X}^{i_1}_t, \cdots, \bar{X}^{i_l}_t, -\tau\leq t \leq T\wedge \tau_U)$. This property implies the convergence in law of $(X^{i_1,N}_t, \cdots, X^{i_l,N}_t, -\tau\leq t \leq T)$ towards $(\bar{X}^{i_1}_t, \cdots, \bar{X}^{i_l}_t, -\tau\leq t \leq T)$ whose law is equal to $m^{p(i_1)}\otimes \cdots \otimes m^{p(i_l)}$.
\end{proof}

These two results readily yields the announced results on the quenched propagation of chaos and convergence for almost all realization of the delays in the translation invariant case. Let us now deal with the averaged convergence of the process in the non translation invariant case.

\subsection{Averaged propagation of chaos in the non-translation invariant case}
We now consider the case of non-translation invariant neural fields. In that case, the distribution of delays depends on the properties of the particular neuron considered (see e.g. the case of neurons in a box treated in section~\ref{sec:FiringRate}). Let us consider neuron $i$ in population $\alpha$. The delays $\tau_{ij}$ between neuron $i$ and neuron $j$ of population $\gamma$ has a distribution given by $\eta_{i\gamma}$ and these different laws, averaged across all possible choices of location of neuron $i$, has the law given by $\eta_{\alpha\gamma}$.

In the previous section, we were able to take into account possible correlations between delays: indeed, the only property used was that for any fixed neuron, the sequence of delays $(\tau_{ij})_{j\in\{1\cdots N\}}$ were independent, whatever possible correlations between $\tau_{ij}$ and $\tau_{kl}$ for $k\neq i$. We make the same assumption here. We denote by $\E_i$ the expectation over all possible choices of locations for neuron $i$, i.e. all possible distributions $\eta_{i\gamma}$. Following the same lines as done in the quenched heterogeneity case for translation invariant systems, we show the following:

\begin{theorem}\label{thm:AveragedPropaChaos}
		Under the assumptions (H1)-(H4) and the chaotic initial condition assumption, the process $(X^{i,N}_t, -\tau \leq t\leq T)$, solution of the network equations~\eqref{eq:Network} averaged across all possible values of the delays converge towards the process $(\bar{X}^i_t, -\tau \leq t\leq T)$ solution of the mean-field equations~\eqref{eq:Coupling}. This implies in particular convergence in law of the process $(\E_i[X^{i,N}_t], -\tau \leq t\leq T)$ towards $(\bar{X}^{\alpha}_t, -\tau \leq t\leq T)$ solution of the mean-field equations~\eqref{eq:MFE}. Moreover, since $f$ and $g$ are globally Lipschitz-continuous, we have for any $i\in \N$ and any $T>0$:
	\begin{equation}\label{eq:PropchaosAveraged}
		\max_{i=1\cdots N} \E\left\{\Exp\Big [\sup_{-\tau \leq s\leq T} \vert \E_i[X^{i,N}_s] - \bar{X}^i_s\vert^2 \Big]\right\}< \frac{C(T)}{\min_{\gamma}(N_{\gamma})}
	\end{equation}
	where $C(T)$ is a constant only depending on the parameters of the system and the time horizon $T$.
\end{theorem}

\begin{proof}
The proof proceeds as the one of theorem~\ref{thm:PropagationChaos} by thoroughly controlling the distance between $\E_i[X^{i,N}_t]$ and $\bar{X}^i_t$ under the norm $\Vert Z\Vert^2=\Exp[\sup_{-\tau\leq t\leq T} \vert Z_s \vert^2]$. This distance is decomposed into the sum of eight terms analogous to the quenched translation invariant case, that are controlled exactly in the same manner as done in the proof of theorem~\ref{thm:PropagationChaos}, except two terms $E'_t(N)$ and $H'_t(N)$ analogous to $E_t(N)$ and $H_t(N)$. Regarding these terms, we have:
\begin{align}
	\nonumber\E\{\Exp[\sup_{0\leq s\leq t}\vert E_s'(N) \vert ^2]\} &=\E\Big\{\Exp[\sup_{0\leq s\leq t} \vert \int_0^s \sum_{\gamma=1}^P  \frac 1 {N_{\gamma}} \sum_{p(j)=\gamma} \E_i\big\{(b_{\alpha\gamma}(\bar{X}^i_s,\bar{X}^j_{s-\tau_{ij}})\big\}-\int_{-\tau}^0\Exp_Z[b_{\alpha\gamma}(\bar{X}^i_s,Z_{s+u}^{\gamma})] ) \,d\eta_{\alpha\gamma}(u)\, ds\vert^2]\Big\}\\
	\label{eq:EtHtPrimeControl}& \leq T P  \sum_{\gamma=1}^P \int_0^t \E\Big\{ \Exp[\bigg \vert \frac 1 {N_{\gamma}} \sum_{p(j)=\gamma} \E_i\big\{b_{\alpha\gamma}(\bar{X}^i_s,\bar{X}^j_{s-\tau_{ij}})\big\}- \int_{-\tau}^0\Exp_Z[b_{\alpha\gamma}(\bar{X}^i_s,Z_{s+u}^{\gamma})\, d\eta_{\alpha\gamma}(u)] \bigg\vert^2] \, ds\Big\}
\end{align}
and using Burkholder-Davis-Gundy
\[
	\E\Big\{\Exp[\sup_{s\leq t}\vert H'_s(N) \vert ^2]\Big\} \leq 4 P  \sum_{\gamma=1}^P \int_0^t  \E\Big\{\Exp[\bigg \vert \frac 1 {N_{\gamma}} \sum_{p(j)=\gamma} \E_i\big\{\beta_{\alpha\gamma}(\bar{X}^i_s,\bar{X}^j_{s-\tau_{ij}})\big\}- \int_{-\tau}^0\Exp_Z[\beta_{\alpha\gamma}(\bar{X}^i_s,Z_{s+u}^{\gamma}) \, d\eta_{\alpha\gamma}(u)] \bigg\vert^2]\Big\}\, ds.\]
	In that case again, the expectation of the process $\E_i\big\{\beta_{\alpha\gamma}(\bar{X}^i_s,\bar{X}^j_{s-\tau_{ij}})\big\}$ under $\bar{X}^j$ and the delays is $\int_{-\tau}^0\Exp_Z[\beta_{\alpha\gamma}(\bar{X}^i_s,Z_{s+u}^{\gamma}) \, d\eta_{\alpha\gamma}(u)]$, and hence developing this expression, we are left with at most $3N_{\gamma}$ terms that are bounded, and we can hence conclude on the convergence of the averaged law towards the mean-field equations, and obtain the announced speed of convergence.
\end{proof}

Let us now go deeper into the dynamics of these limit equations.

\section{Dynamics of the firing-rate model}\label{sec:FiringRate}
In the previous sections, we showed the convergence of network equations~\eqref{eq:Network} towards well-posed delayed McKean-Vlasov equations~\eqref{eq:MFE}. These equations form a system of implicit equations in the space of stochastic processes. Understanding the qualitative dynamics of the networks in the limit $N\to\infty$ is hence highly challenging in this general context. In order to go deeper into the analysis of the dynamics of the mean-field equations, we consider the dynamics of the classical firing-rate model.

\subsection{Mean-Field equations for the firing-rate model}
We now instantiate the dynamics of our neurons as firing-rate models. In this model the intrinsic dynamics $f_{\alpha}(x,t)=-x/\theta_{\alpha}+I_{\alpha}(t)$ is affine, the diffusion function $g_{\alpha}(x,t)$ is constant (equal to $\lambda_{\alpha}$) and in which the interactions only depend on a nonlinear transform of the membrane potential of presynaptic neurons: $b_{\alpha\gamma}(x,y)=\bar{J}_{\alpha\gamma} S(y)$ and $\beta_{\alpha\gamma}=\sigma_{\alpha\gamma}S(y)$. It is obvious that the assumptions of the above sections are satisfied by the firing-rate model. Therefore, when considering the delays $\eta_{i\gamma}$ only depending on $p(i)$, we have propagation of chaos and almost sure (quenched) convergence towards the mean-field equations:
\begin{multline}\label{eq:MFEWC}
	d\bar{X}^{\alpha}_t= \Big(-\frac{\bar{X}^{\alpha}_t}{\theta_{\alpha}} + I_{\alpha}(t) +  \sum_{\gamma=1}^P \bar{J}_{\alpha\gamma}\int_{-\tau}^0  \Exp_{\bar{X}}[S(\bar{X}^{\gamma}_{t+s})]d\eta_{\alpha\gamma}(s) \Big)\,dt \\
	+ \lambda_{\alpha}\,dW_t^{\alpha}
	+ \sum_{\gamma=1}^P\sigma_{\alpha\gamma}\int_{-\tau}^0 \Exp_{\bar{X}}[S(\bar{X}^{\gamma}_{t+s})]d\eta_{\alpha\gamma}(s) \,dB^{\alpha\gamma}_t
\end{multline}
and if the delay distributions $\eta_{i\gamma}$ depend on the choice of $i$, the same result hold in an averaged sense.

In these cases, we can reduce the dynamics as follows:
\begin{theorem}\label{thm:FiringRateGaussian}
	In the firing-rate model, solutions of the mean-field equations ~\eqref{eq:MFEWC} are exponentially fast attracted towards Gaussian solutions. When the initial conditions are Gaussian processes, the solution is Gaussian, with mean $\mu_{\alpha}$ and variance $v_{\alpha}$ satisfying a well-posed system of delayed differential equations:
	\begin{equation}\label{eq:DDE}
		\begin{cases}
			\dot{\mu}_{\alpha} &= -\frac{1}{\theta_{\alpha}}\mu_{\alpha}+\sum_{\gamma=1}^P\int_{-\tau}^0 f(\mu_{\gamma}(t+s),v_{\gamma}(t+s))d\eta_{\alpha\gamma}(s) + I_{\alpha}(t)\\
			\dot{v}_{\alpha} &= -\frac{2}{\theta_{\alpha}}v_{\alpha}+\sum_{\gamma=1}^P \sigma_{\alpha\gamma}^2\left(\int_{-\tau}^0 f(\mu_{\gamma}(t+s),v_{\gamma}(t+s))d\eta_{\alpha\gamma}(s)\right)^2 + \lambda_{\alpha}^2\\
		\end{cases}
	\end{equation}
	where $f(\mu,v)=\Exp[S(X)]$ for $X\sim \mathcal{N}(\mu,v)$.
\end{theorem}

\begin{proof}
	By the variation of constant formula we have:
	\begin{multline*}
		\bar{X}^{\alpha}_t= X^{\alpha}_0e^{-t/\theta_{\alpha}} + \int_0^t e^{-(t-s)/\theta_{\alpha}} \Big(I_{\alpha}(s) +  \sum_{\gamma=1}^P \bar{J}_{\alpha\gamma}\int_{-\tau}^0  \Exp_{\bar{X}}[S(\bar{X}^{\gamma}_{s+u})]d\eta_{\alpha\gamma}(u) \Big)\,ds \\
		+ \int_0^t e^{-(t-s)/\theta_{\alpha}} \lambda_{\alpha}\,dW_s^{\alpha}
		+  \sum_{\gamma=1}^P\sigma_{\alpha\gamma}\int_0^t e^{-(t-s)/\theta_{\alpha}} \int_{-\tau}^0 \Exp_{\bar{X}}[S(\bar{X}^{\gamma}_{s+u})]d\eta_{\alpha\gamma}(u) \,dB^{\alpha\gamma}_s.
	\end{multline*}
	The initial condition term $X^{\alpha}_0e^{-t/\theta_{\alpha}}$ vanishes exponentially fast. The other terms include deterministic terms and stochastic integrals with respect to Brownian motions of deterministic functions, hence are Gaussian. All but the initial condition term are hence Gaussian and the solutions are exponentially fast attracted towards Gaussian solutions. Taking the mean and covariance function of this process readily yields equations~\eqref{eq:DDE}.
\end{proof}

In the firing-rate case, we hence have an important reduction of complexity. We now exploit this simpler form of the mean-field equations to analyze the dynamics of the network using the theory of delayed differential equations (see e.g.~\cite{hale-lunel:93}) and compare it to numerical simulations. We particularly focus on the importance of the noise in the stochastic weights and on the distribution of delays.

\subsection{Distributed delays in a one-population network}
In order to analyze the impact of the distribution of delays on the dynamics of the network, we simplify the model considering a one-population network ($P=1$, we hence drop the population indexes since there is no ambiguity here). We further assume that $\mu=0$, or more precisely consider centered sigmoid functions ($S(0)=0$) and no input ($I=0$). In order to further simplify the system, we consider $S(x)=\erf(gx)$ where $\erf(x)=\frac{1}{\sqrt{2\pi}}\int_0^x e^{-x^2/2}$. In that case, integration by parts yields:
\[f(\mu,v)=\erf(\frac{g\mu}{\sqrt{1+g^2v}}).\]
In that simplified case, it is easy to see that a stationary solution is given by $(\mu^*,v^*)=(0,\lambda^2/2)$. Due to the quadratic form of the variance equation, the linearized equation on the variance is trivial $\dot{x}_{\alpha} = -\frac{2}{\theta_{\alpha}}x_{\alpha}$. Therefore, the level of noise in the connections do not come into play in that particular case, and the stability of the fixed point only depends on the characteristic equation related to the main, which is given by the dispersion relationship:
\[\xi=-\frac{1}{\theta} + \frac{{J}g}{ \sqrt{2\pi(1+g^2\lambda^2/2)}}\int_{-\tau}^0e^{\xi s}d\eta(s). \]
The solutions of this equations are the characteristic exponents of the system, and relate to the stability of the fixed point considered: if all characteristic exponents have negative real part, the equilibrium is asymptotically exponentially stable, and there exists a characteristic exponent with strictly positive real part, the equilibrium is unstable. Changes in the sign of the real part of the characteristic exponent constitute bifurcations of the system. Let us start by analyzing the presence of \emph{pitchfork} bifurcations that correspond to real characteristic exponents crossing the imaginary axis. Such bifurcations hence occur when the following relationship is satisfied:
\[0=-\frac{1}{\theta} + \frac{{J}g}{ \sqrt{2\pi(1+g^2\lambda^2/2)}} \]
and hence this bifurcation occurs independently of the choice of the delay. Turing-Hopf bifurcations occur when the system has a pair of complex conjugate characteristic exponents with non-zero imaginary part crossing the imaginary axis. Such bifurcations hence occur when there exists $\omega>0$ such that $\xi=\mathbf{i}\omega$, i.e.:
\[\mathbf{i}\omega=-\frac{1}{\theta} +\frac{{J}g}{ \sqrt{2\pi(1+g^2\lambda^2/2)}}\int_{-\tau}^0e^{\mathbf{i}\omega s}d\eta(s).\]
The existence of such bifurcations exist depending on the type of delay distribution. It is important to note at this point that there is no such bifurcation when delays are equal to zero.

\subsubsection{Dirac delta delays}
Assuming that the neuronal populations are highly concentrated so that delays can be considered constant (i.e. $\eta=\delta_{-\tau}$). The dispersion relationship can be easily solve and one finds that Hopf bifurcations exist only for
\[ \frac{\vert J\vert g}{\sqrt{2\pi(1+g^2\frac{\lambda^2}{2})}} >1,\]
i.e. for strong enough connectivity and small enough noise. Under that assumption is then easy to show that Turing-Hopf bifurcations arise when the parameters satisfy the relationship, when $J<0$:
\[
	\tau = \pi-\frac{1}{\omega} \arctan(\omega \theta)
\]
with $	\omega =\sqrt{\frac{J^2g^2}{2\pi(1+g^2\frac{\lambda^2}{2})} - 1}$. These result into oscillations of the solutions at a pulsation equal by $\omega$. Details of the calculations are left to the reader and follow the lines of the proof used for localized delays of the next section. Note that the value of the delay corresponding to the Hopf bifurcation is a function of $\lambda$. In particular, no Hopf  bifurcation is to be found for $\lambda > \lambda^*=\sqrt{2(J^2/(2\pi)-1/g^2)}$. The curve of Hopf bifurcations obtained in the plane $(\lambda,\tau)$ and simulations the moment equations are displayed in figure Fig.~\ref{fig:DeltaDelay}.
\begin{figure}[htbp]
	\centering
		\subfigure[Bifurcation diagram]{\includegraphics[width=.45\textwidth]{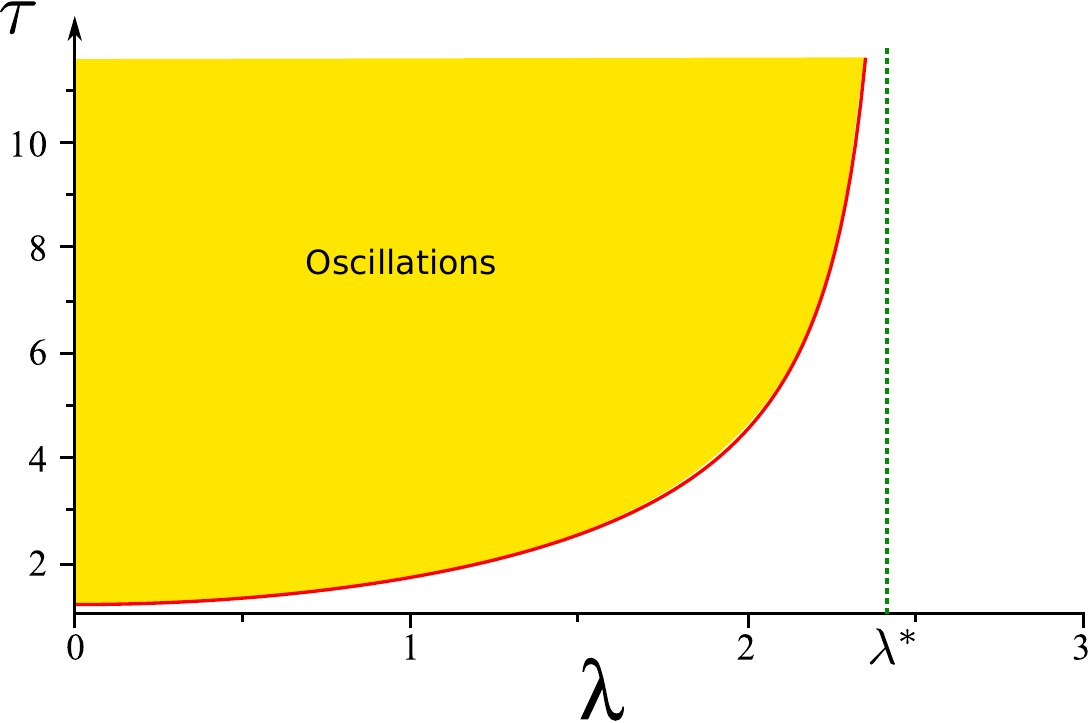}}
		\subfigure[$\lambda=0.5$, $\tau=1$]{\includegraphics[width=.45\textwidth]{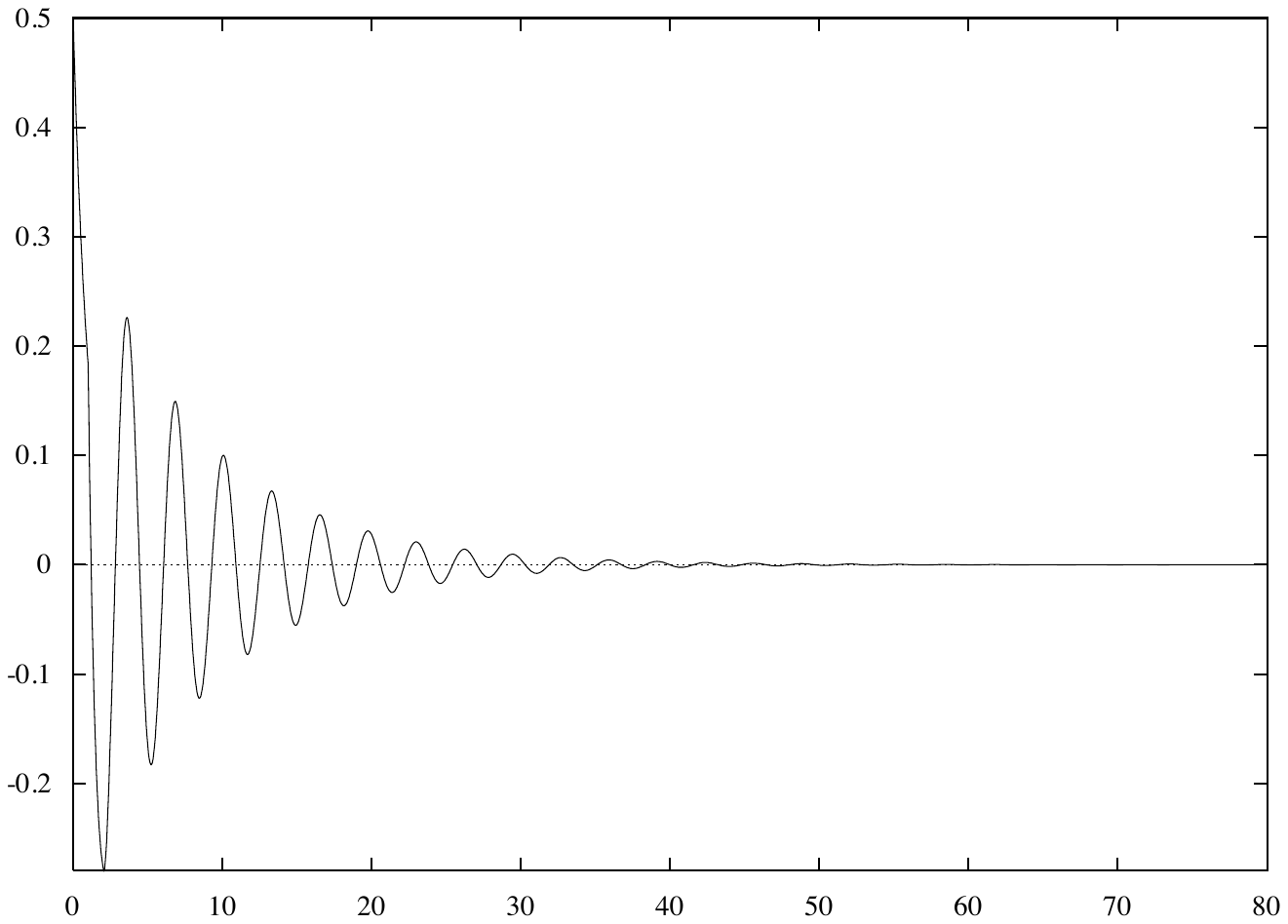}}
		\subfigure[$\lambda=0.5$, $\tau=1.5$]{\includegraphics[width=.45\textwidth]{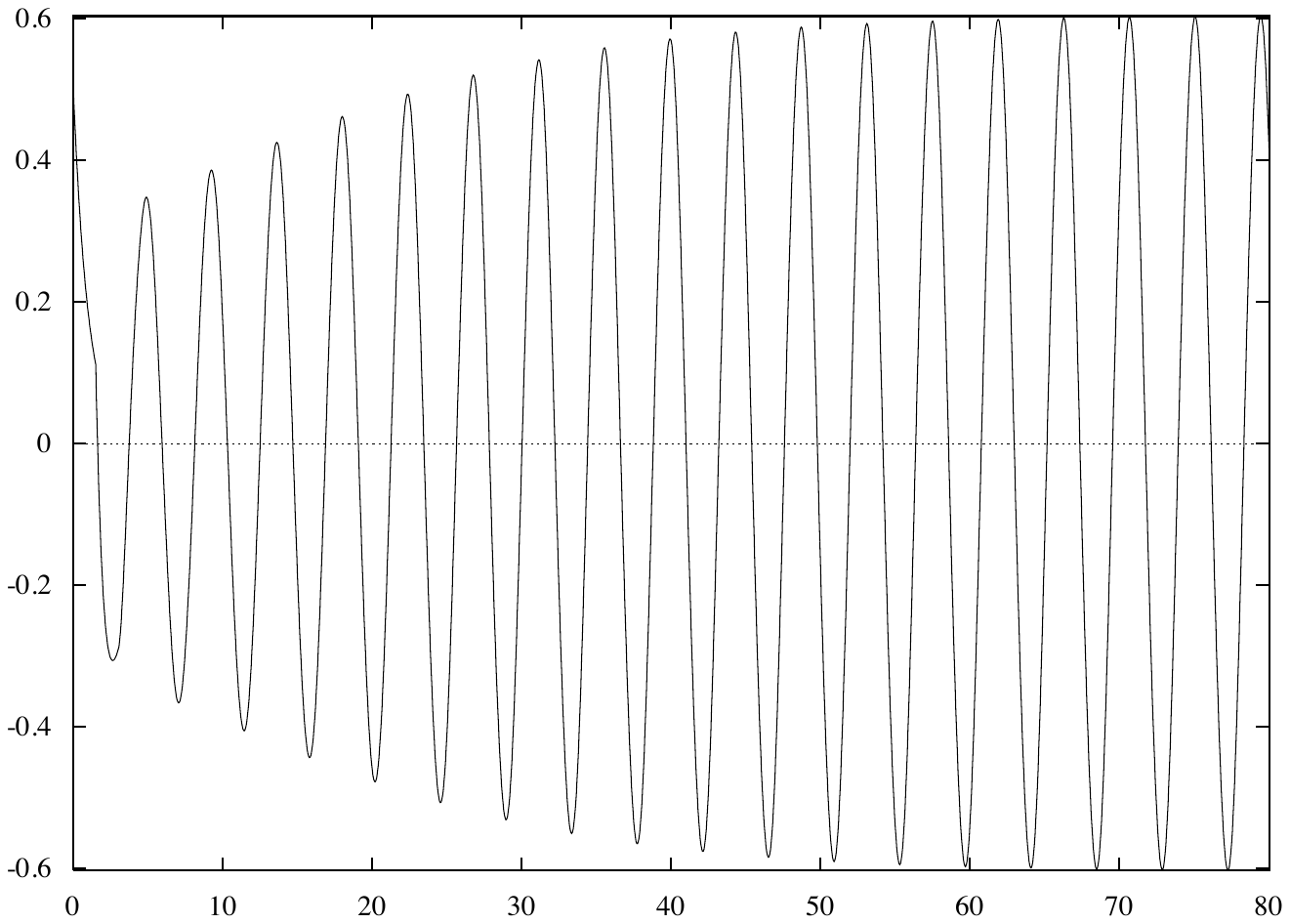}}
		\subfigure[$\lambda=1$, $\tau=1.5$]{\includegraphics[width=.45\textwidth]{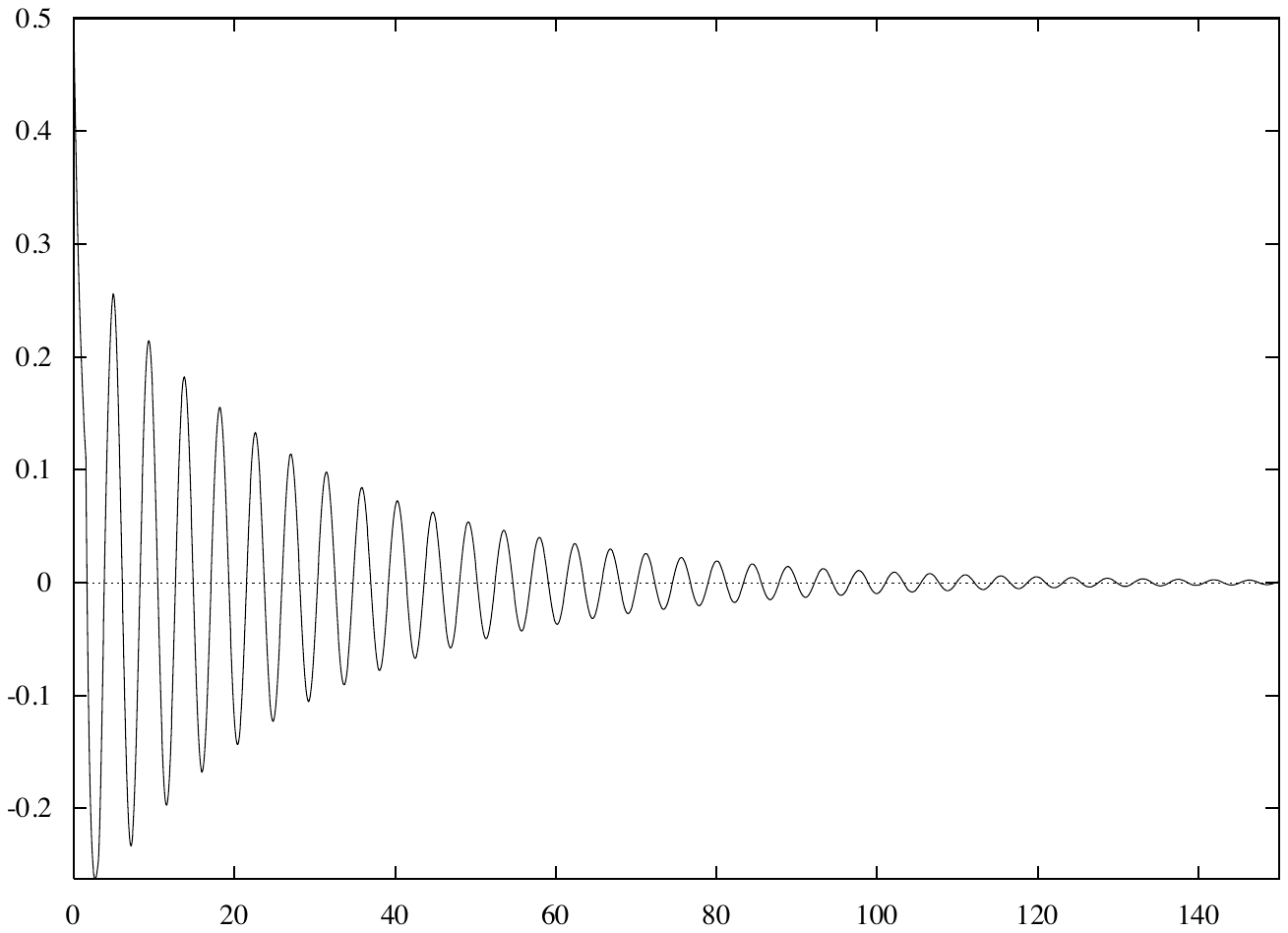}}
	\caption{Dirac delta distributed delays, $g=1$, $J=-2$, $\lambda^*\simeq 2.449$. For $\lambda=0.5$, the Hopf bifurcation occurs for a value of $\tau\simeq 1.33$ and for $\lambda=1$, the Hopf bifurcation occurs for $\tau\simeq 1.727$. }
	\label{fig:DeltaDelay}
\end{figure}
These simulations illustrate the fact that, depending on the noise and the delay chosen, the law of the mean-field equations has a mean that is either periodic or constant, yielding periodic or stationary laws of the mean-field equations. Since all neurons have the same probability distribution, in the periodic case all neurons will oscillate in phase.

{We now compare these analytical results to numerical simulations of the initial stochastic finite-size network given by equations~\eqref{eq:Network}. The trajectories of the reduced delayed dynamical system~\eqref{eq:DDE} were obtained using XPPAut software~\cite{ermentrout:02} that uses a classical fourth order Runge-Kutta method with adaptive timestep, and the stochastic trajectories of the network were obtained by direct numerical integration of equations~\eqref{eq:Network} using a vectorized implementation of the Euler-Maruyama integration scheme with time step $0.005$. These results are displayed in figure Fig.~\ref{fig:NetworkDeltaDelay}}.
\begin{figure}
	\centering
		\subfigure[$\lambda=0.5$, $\tau=1$]{\includegraphics[width=.3\textwidth]{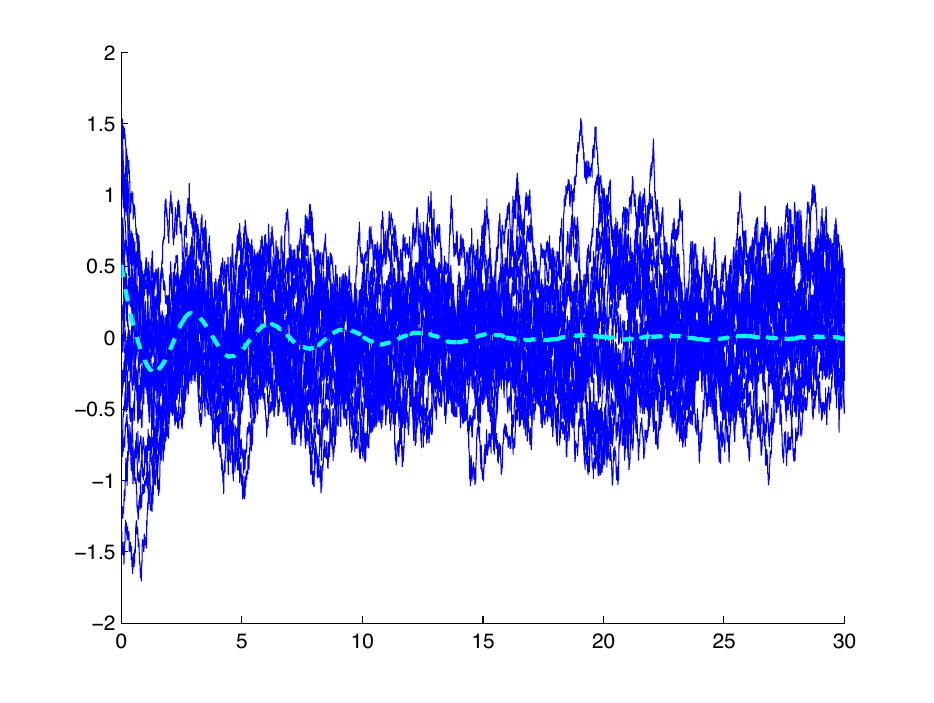}}
		\subfigure[$\lambda=0.5$, $\tau=1.5$]{\includegraphics[width=.3\textwidth]{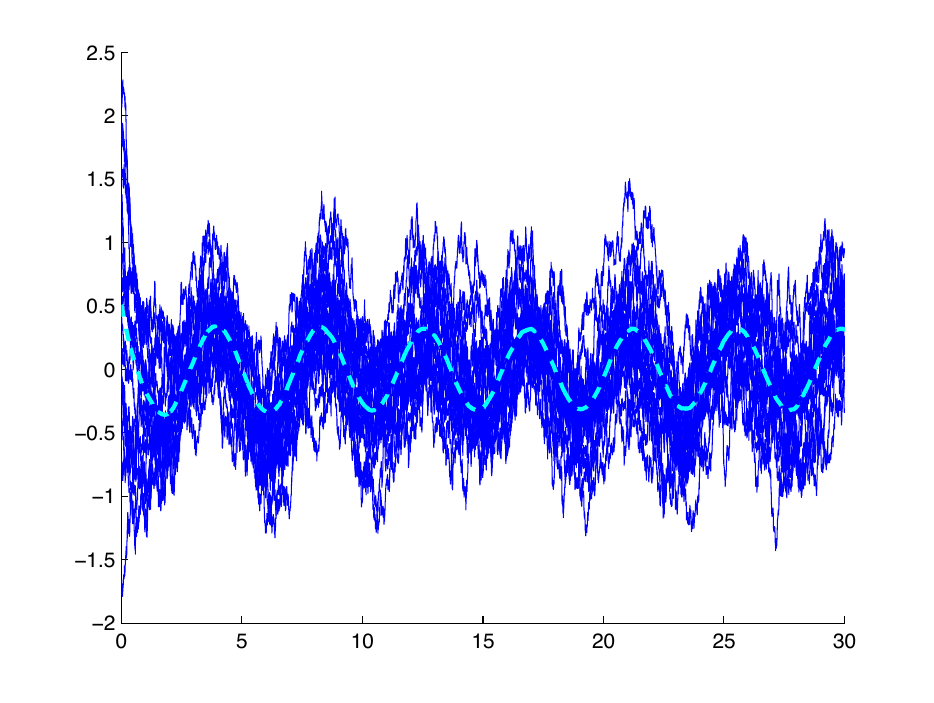}}
		\subfigure[$\lambda=1$, $\tau=1.5$]{\includegraphics[width=.3\textwidth]{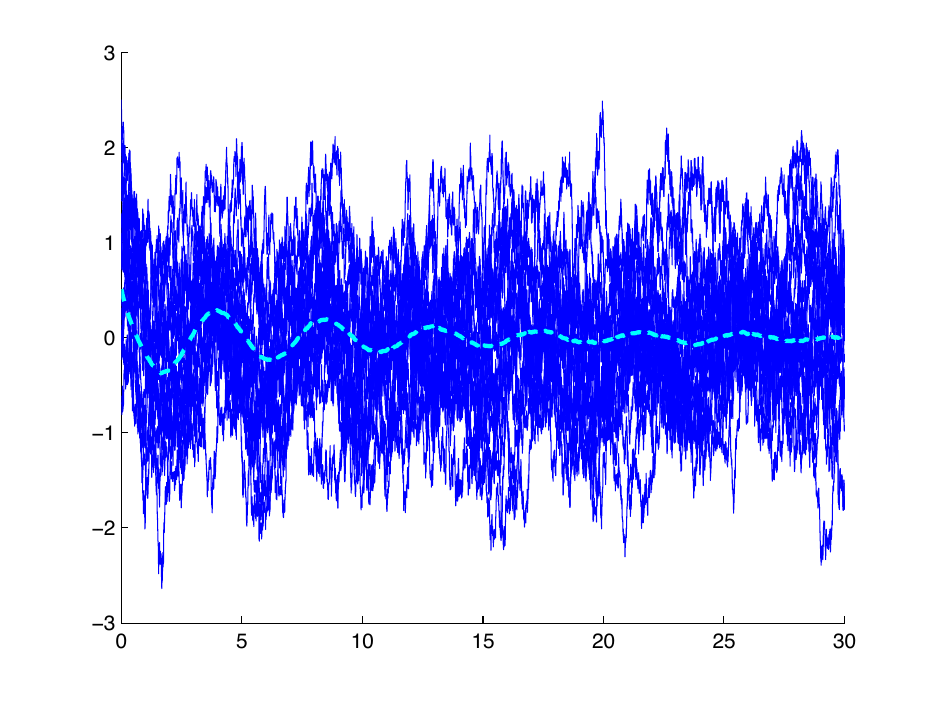}}
	\caption{Simulation of a $3\,000$ neurons network with Dirac delta distributed delays corresponding to the cases presented in Fig.~\ref{fig:DeltaDelay}. blue: $30$ arbitrarily chosen trajectories in the network, and cyan dashed line: empirical average. The dynamics observed for the mean-field equations is clearly recovered in the network simulations.}
	\label{fig:NetworkDeltaDelay}
\end{figure}

\subsubsection{Localized uniformly distributed delays}
We now consider instead of Dirac delta distributions of the delays that the neurons in each population have a dispersion of order $\delta$ around a typical value $\tau$. One of the simplest distribution modeling this kind of dispersion is the uniform distribution $\eta=\mathbbm{1}_{[\tau-\delta/2,\tau+\delta/2]} dt/\delta$. The Dirac delta delay can be seen as the limit of this distribution as $\delta \to 0$. It is easy to show that we have in that case a dispersion relationship:
\[\xi = -\frac{1}{\theta} + \frac{{J}g}{ \sqrt{2\pi(1+g^2\lambda^2/2)}} e^{-\xi\tau} \frac{e^{\xi\delta/2}-e^{-\xi\delta/2}}{\xi\delta}\]
and possible Hopf bifurcations occurring when there exists $\omega>0$ such that:
\[\mathbf{i}\omega = -\frac{1}{\theta} + 2\frac{{J}g}{ \sqrt{2\pi(1+g^2\lambda^2/2)}} e^{-\mathbf{i}\omega\tau} \frac{\sin({\omega\delta/2})}{\omega\delta}\]
Equating the real parts and the imaginary parts on both sides yields the system:
\[\begin{cases}
	\frac 1 \theta &= \frac{{J}g}{ \sqrt{2\pi(1+g^2\lambda^2/2)}} \frac{\sin(\Omega)}{\Omega} \cos(\omega\tau)\\
	\omega &= -  \frac{{J}g}{ \sqrt{2\pi(1+g^2\lambda^2/2)}} \frac{\sin(\Omega)}{\Omega}\sin(\omega\tau)
\end{cases}\]
where we denoted $\Omega=\omega\delta/2$. These equations are relatively hard to solve. But since we are interested in parameters $(\tau,\delta)$ corresponding to Hopf bifurcations, we can consider that these equations are self-consistent equations in the plane $(\tau,\delta)$ parametrized by the variable $\omega$ (or more precisely $\Omega$) and obtain the following equations on the Hopf bifurcation curves when $J<0$:
\[\begin{cases}
	\delta &= \displaystyle{\left(\frac 1 {4\Omega^2} \left(-\frac{1}{\theta^2} + \frac{{J}^2g^2}{ {2\pi(1+g^2\lambda^2/2)}}\frac{\sin^2(\Omega)}{\Omega^2}\right)\right)^{-\frac 1 2}}\\
	\\
	\tau &= \displaystyle{\frac{\pi-\arctan(\omega\theta)}{\omega}= \frac{\pi-\arctan(2\Omega \theta/\delta)}{2\Omega/\delta}}
\end{cases}\]

\begin{figure}[!h]
	\centering
		\subfigure[Bifurcation diagram $J=-2$, $\theta=1$]{\includegraphics[width=.25\textwidth]{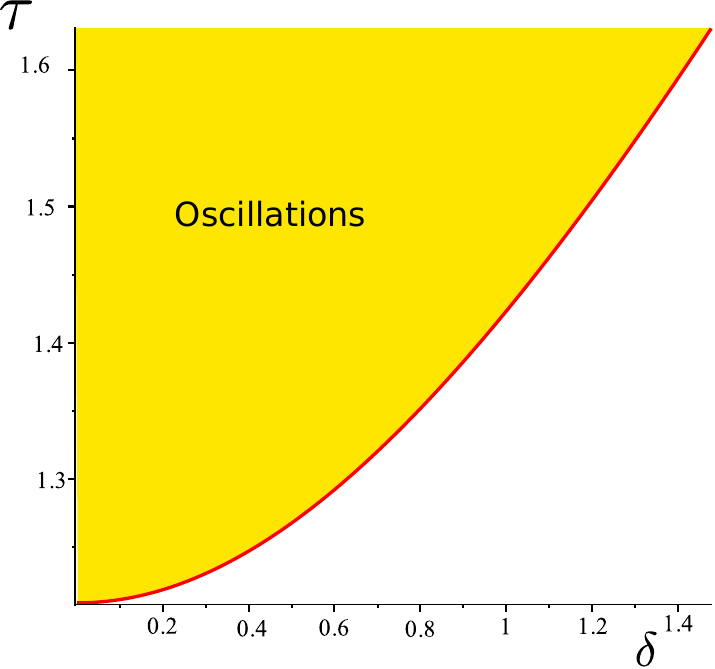}}\qquad
		\subfigure[Bifurcation diagram $J=-5$, $\theta=2$]{\includegraphics[width=.25\textwidth]{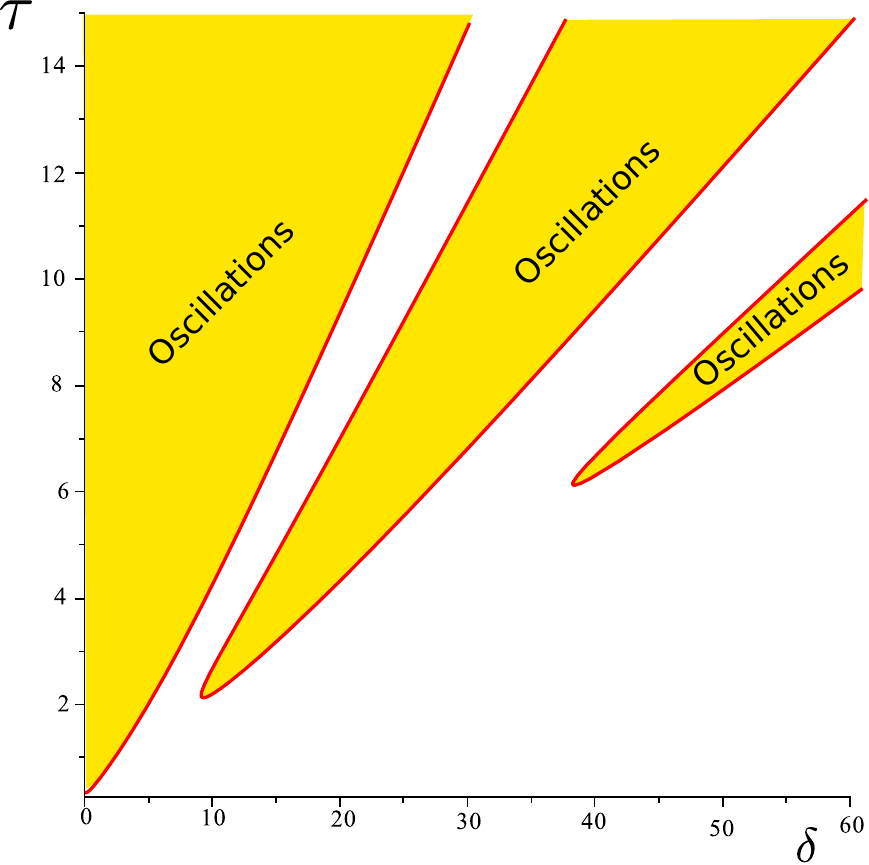}}\\
		\subfigure[$\tau=1.5$, $\delta=0.6$]{\includegraphics[width=.3\textwidth]{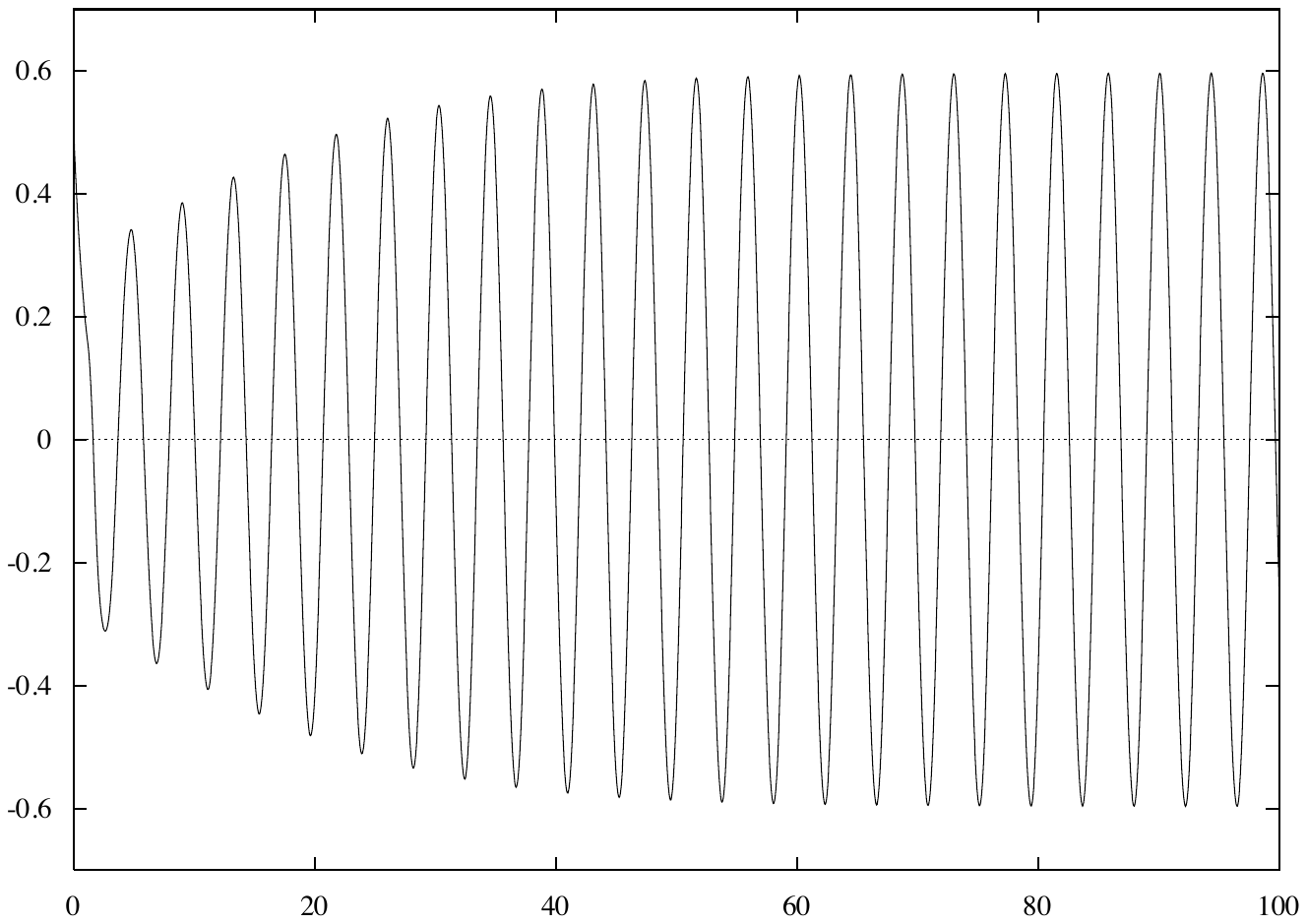}}
				\subfigure[$\tau=1.5$, $\delta=0.6$, network]{\includegraphics[width=.3\textwidth]{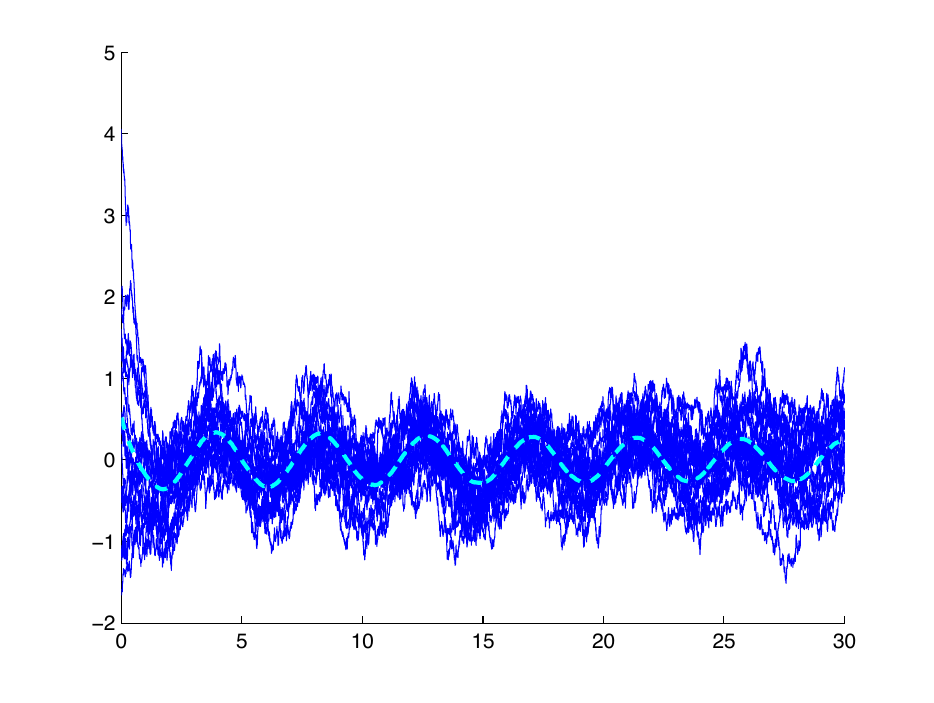}}\\
		\subfigure[$\tau=1.5$, $\delta=2$]{\includegraphics[width=.3\textwidth]{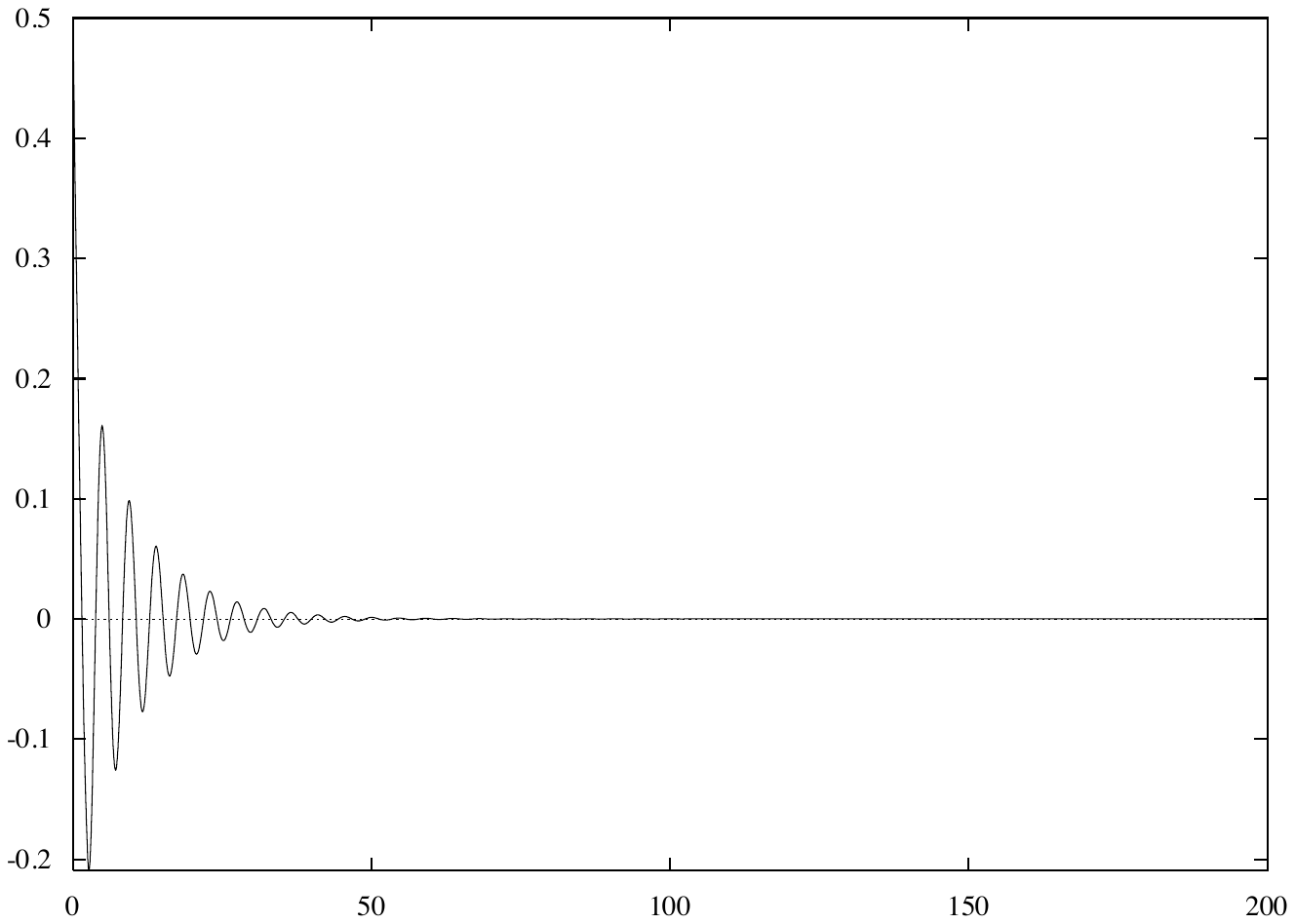}}
		\subfigure[$\tau=1.5$, $\delta=2$, network]{\includegraphics[width=.3\textwidth]{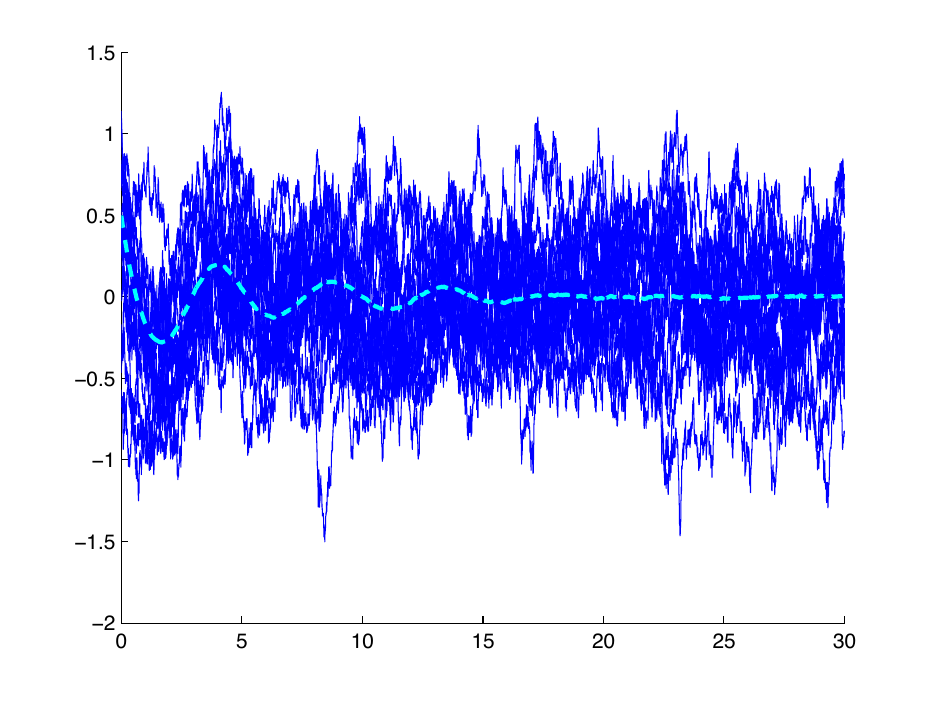}}
	\caption{Uniformly distributed delays in $[\tau-\delta/2,\tau+\delta/2]$ for $J=-2$ and $\theta=1$ (all but (b)) or $J=-5$ and $\theta=2$ (b), $g=1$ and $\lambda=1$. (a-b): Hopf bifurcations as a function of the spread of the distribution and the delay. Strange phenomena with multiple Hopf bifurcations can arise as $J$ is decreased (b). (c-f): for $\tau=1.5$, oscillations present in the Dirac-delta delayed network persist for small values of $\delta$. However, when $\delta$ exceeds the value related to the Hopf bifurcation, the oscillations disappear.}
	\label{fig:UniformDelays}
\end{figure}

These curves are plotted for different choices of parameters in figure Fig.~\ref{fig:UniformDelays}. For small absolute values of the connectivity coefficient $J$, the system presents a Hopf bifurcation curve which is an increasing function of $\delta$: the larger the dispersion, the larger the necessary delay to trigger oscillations. For instance fixing $J=-2$ and $\theta=1$ and a fixed value of the delay $\tau=1.5$, we show that increasing the dispersion of the delays $\delta$ destroys the oscillatory pattern: the network shows a transition from perfectly periodic in law synchronized oscillations to stationary asynchronous behaviors. For a different set parameters, we display a relatively complex bifurcation diagram presenting several parameter regions corresponding to the first Hopf bifurcation and hence to the presence of oscillations.

\subsubsection{Neurons in an interval}
Let us now further specify the type of distribution of delays that arises in neurosciences. As stated, delays are chiefly related to the axonal propagation at finite speed and to the specific time the signal transmission takes, $\tau_s$. Assuming constant speed $c$, the delay $\tau_{ij}$ between neuron $i$ at location $r_i$ and a neuron $j$ at location $r_j$ is $\tau_{ij} = \frac{\vert r_i-r_j \vert}{c} + \tau_s$. Delays are hence strongly related to the distribution of neurons in the cortex.  Moreover, in that case, the system may not be translation invariant (if we do not assume for instance periodicity) as we will see below. We hence need to identify the distribution of distances between distributed points in specific topologies. Such distributions are known in closed form for some simple topologies, such as the rectangle or the ball. But even in these favorable cases, the bifurcations appear harder to characterize.

This is however possible rigorously in a 1-dimensional case. Assuming uniform distribution of the neurons in an interval $[0,a]$ and that $c=1$ (without loss of generality), it is easy to show that the distribution of the distance $r$ between neuron $i$ at location $r_i$ and other neurons has the density:
\[\tilde{\eta}_i=\left(\mathbbm{1}_{[0,a-r_i]}+ \mathbbm{1}_{[0,r_i]}\right)\frac{dr}{a}\]
and the averaged law the density:
\[d\tilde{\eta}(r)=\left(\frac{2}{a}-\frac{2r}{a^2}\right)dr.\]

In the case where we consider a periodic neural area (i.e. identifying $0$ and $a$), the distribution of distances is clearly the uniform distribution on $[0,a/2]$ and the results of the previous section apply.

When considering non-periodic neural areas, applying the results of the above section ensures that an averaged propagation of chaos takes place, and that the network equations converge towards the mean-field equations given by theorem~\ref{thm:FiringRateGaussian}. A Gaussian process with mean $0$ and covariance $\lambda^2/2e^{-(t-s)/\theta}$ is a stationary solution, i.e. the point $(0,\lambda^2/2)$ is a fixed point of the moment equations~\eqref{eq:DDE}. At this point the dispersion relationship now reads, denoting $\bar{J}=\frac{\bar{J}g}{ \sqrt{2\pi(1+g^2\lambda^2/2)}}$:
\[\xi = -\frac{1}{\theta} + \frac{2\bar{J}}{\xi a}\left(1-\frac 1 {\xi a} -\frac{e^{-\xi a}}{\xi a}\right)e^{-\xi\tau_s}.\]
Possible Hopf bifurcations occur when one can find $\omega\in \R$ such that:
\[\mathbf{i}\omega = -\frac{1}{\theta} -\mathbf{i} \frac{2\bar{J}}{\omega a}\left(1+\mathbf{i} \left(\frac 1 {\omega a} +\frac{e^{-\mathbf{i}\omega a}}{\omega a}\right)\right)e^{-\mathbf{i}\omega\tau_s}.\]
In order to solve this equation, we note $\Omega=\omega\,a$ and $Z(\Omega)=\frac{2\bar{J}}{\Omega}\left(1+\mathbf{i} \left(\frac 1 {\Omega} +\frac{e^{-\mathbf{i}\Omega}}{\Omega}\right)\right)$. Remark that $Z(\Omega)$ tends to $1/2$ when $a\to 0$ hence recovering the Dirac delta delay case. Equating modulus and argument in the Hopf dispersion relationship one obtains:
\[\begin{cases}
	\displaystyle{a^2=\frac{\Omega^2}{-\frac{1}{\theta^2}+\vert Z(\Omega)\vert^2}}\\
	\displaystyle{\tau_s=\left(-\frac{\pi}{2} + Arg(Z(\Omega)) -\tan\left(\frac{\Omega\theta}{a}\right)+2k\pi\right)\frac{a}{\Omega}}.
\end{cases}\]
Considering as above $\Omega$ as a parameter, we can find the locus of the Hopf bifurcations in the plane $(a,\tau_s)$. We obtain a sequence of Hopf bifurcations indexed by $k$, and the first bifurcation is responsible for oscillations appearing in the system. Interestingly, we observe that the value of $\tau_s$ corresponding to oscillations takes a minimal value for a certain value of $a$ (see figure Fig.~\ref{fig:NeuronsBox}). This can be understood heuristically as follows. In the interval $[0,a]$, the mean distance between neurons is equal to $\frac a 3$ and the standard deviation to $\frac {a^2} 6$. Increasing $a$ hence has the effect of linearly increasing the averaged effective delay, and quadratically the dispersion. For small values of $a$, the effective delay created by the dispersion of the neuron increases faster than the dispersion, and therefore a smaller constant delay $\tau_s$ is necessary to trigger oscillations. As $a$ is further increased, the dispersion grows faster than the mean effective delay, and in that case the dispersion implies that larger delays are necessary to trigger oscillations, as seen in the uniformly distributed delays case.

\begin{figure}
	\centering
		\subfigure[Bifurcation diagram]{\includegraphics[width=.4\textwidth]{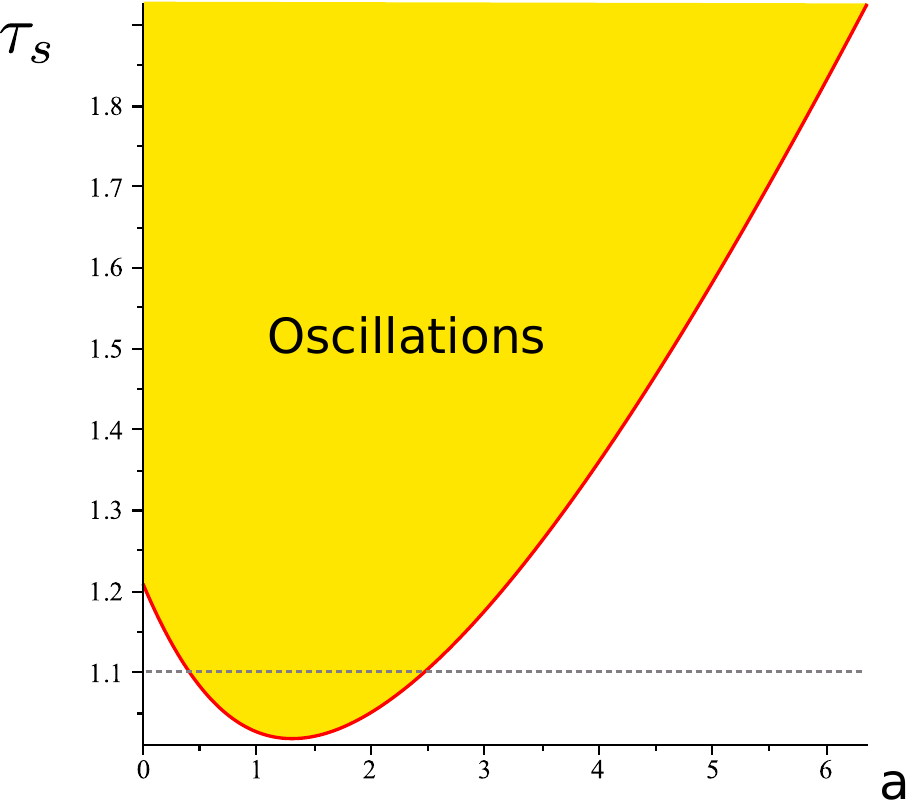}}\\
		\subfigure[$a=0.1$]{\includegraphics[width=.3\textwidth]{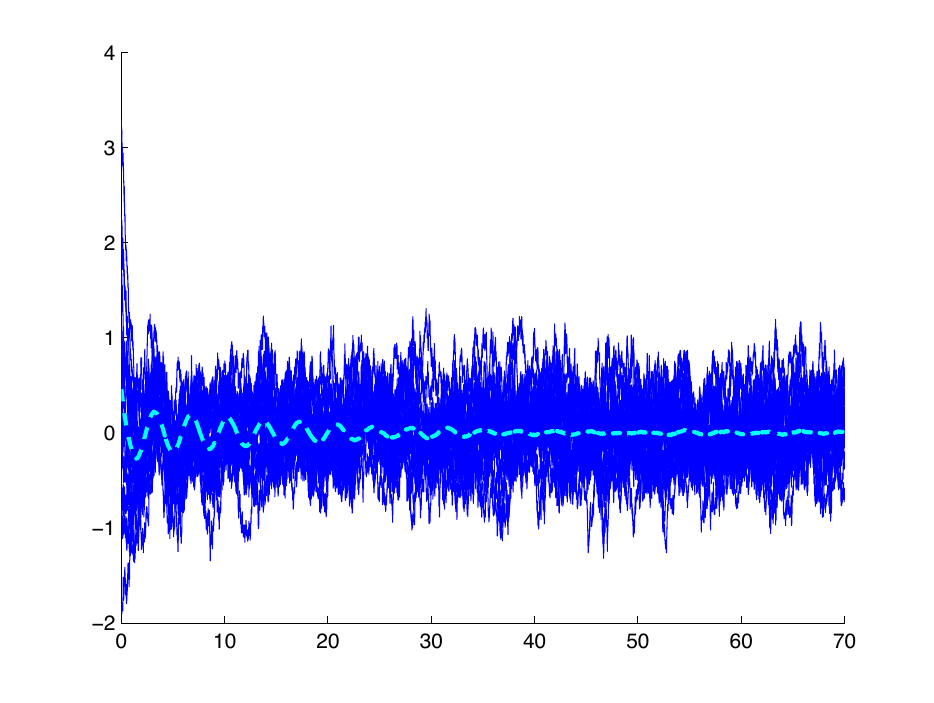}}
		\subfigure[$a=1.5$]{\includegraphics[width=.3\textwidth]{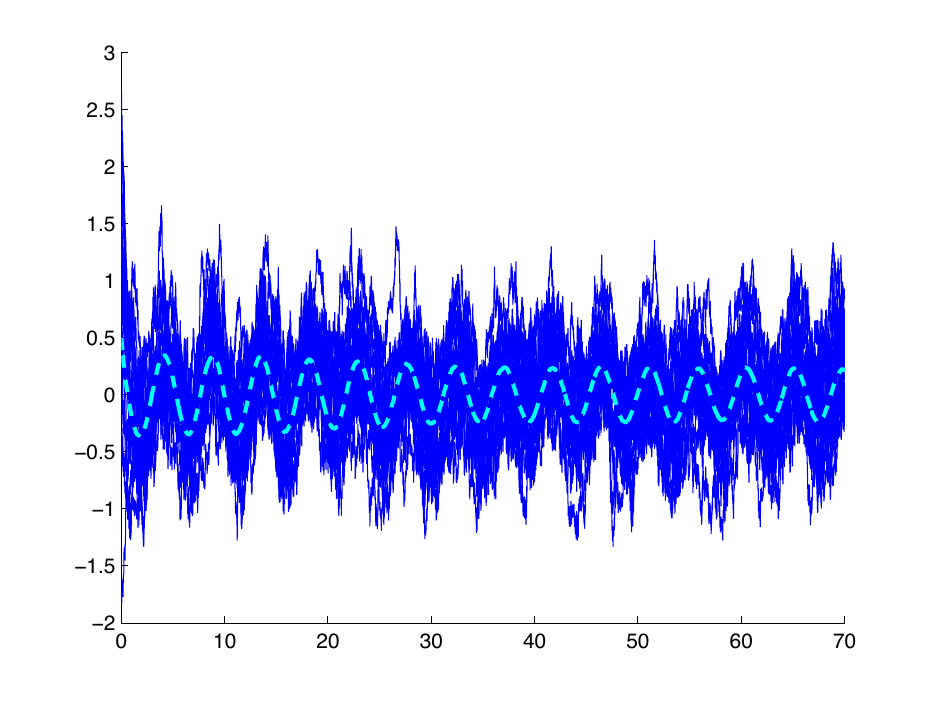}}
		\subfigure[$a=3.5$]{\includegraphics[width=.3\textwidth]{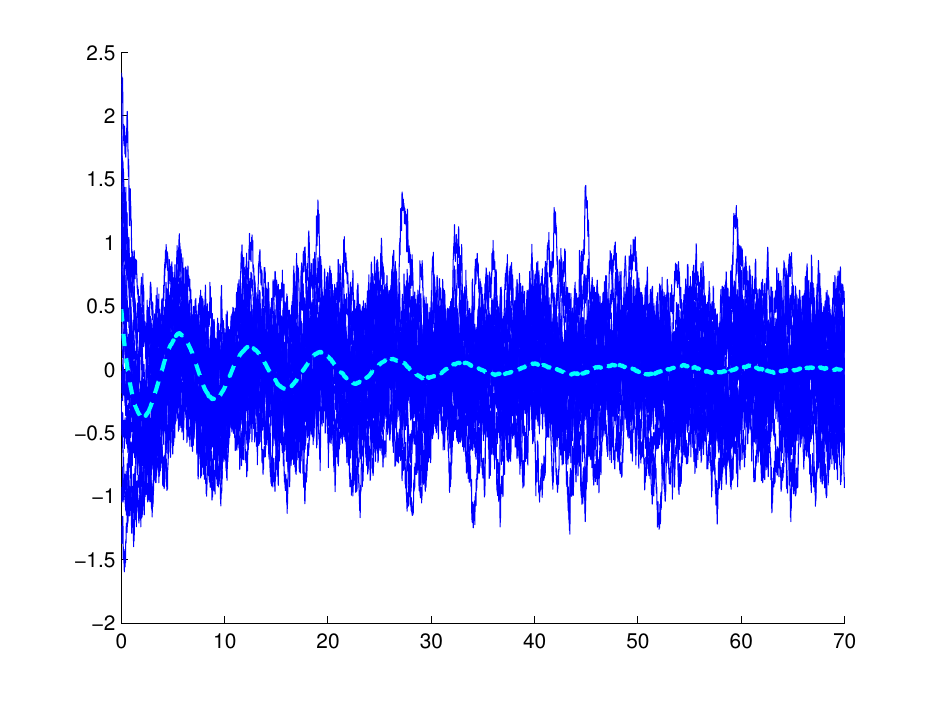}}
	\caption{Neurons uniformly distributed in $[0,a]$ inducing propagation delay and with transmission delay $\tau_s$ : (a) Hopf bifurcation diagram in the plane $(a,\tau_s)$ (same parameters as above). For $\tau_s=1.1$ (b-d), increasing the parameter $a$ (the size of the neural field) induces transitions from stationary to periodic back to stationary.}
	\label{fig:NeuronsBox}
\end{figure}

Fixing all parameters and changing the size of the interval in which the neurons lie can hence induces transition from stationary to periodic and back to stationary. Neurons are hence perfectly synchronized in a specific region of the neural field spatial extension, and otherwise asynchronous.

\section{Conclusion}
In this article, we analyzed the limits and dynamics of large-scale neuronal network models in the presence of noise. We particularly focused on the presence of heterogeneous delays in the interactions between neurons, an ubiquitous phenomenon in neuronal network activity that is understood to have important effects on the brain function. We considered two different kinds of systems: (i) translation invariant systems where the delays have the same distributions whatever the neurons considered, and (ii) non translation invariant case where delays distribution depend on the particular neuron $i$ considered.
In both cases, we derived the mean-field equations associated with such randomly delayed networks for general models, and showed that the propagation of chaos occurred, in a quenched sense for translation invariant systems (i.e. for almost all realization of the delays), or after averaging across possible delay distributions in the non translation invariant case.

In both cases, the limit equations are identical and given by a system of delayed McKean-Vlasov equations with distributed delays. The equations obtained are not classical stochastic equations: they involve a term depending on the law of the solution at previous times (due to the presence of distributed delays). We showed  that these equations were well-posed, and admit a unique solution. Nonetheless, this solution remains relatively hard to characterize at that level of generality. In order to characterize the dynamics of the network and in particular the importance of heterogeneity in the interconnection delays, we instantiated a popular model in neuroscience, the firing-rate model. In the mean-field limit, the solutions converge towards Gaussian processes, whose mean and variance satisfy a closed set of differential equations with distributed delays.

This property allowed, using the theory of bifurcations of delayed differential equations (see e.g.~\cite{hale-lunel:93}), to analyze in detail the role of delays in the macroscopic dynamics of neuronal networks. We showed that depending on the averaged delays, the network can either present a stationary or a synchronized periodic  behavior. Such transitions as a function of the amplitude of delays were already observed in heuristic models including constant delays. But beyond this fact, we showed that averaged delay is not enough to predict the behavior of the system, and a knowledge of the full delay distribution is necessary. We illustrated this fact exhibiting the bifurcation diagram of the Hopf bifurcations as a function of the variance of uniformly distributed delays, and probably more interestingly from the biological viewpoint the dependence of the network behavior as a function of the delays arising from the distance between the cells. We showed that the macroscopic behavior arising from such networks depend on the size of the neural area considered, that parametrizes the delay distributions in the network. This was done analytically in a simple one dimensional case, and observed that the size of the cortical column considered has non-trivial effects on the macroscopic dynamics of the neural area.

This observation is very interesting from the neuroscience viewpoint. Indeed, it has been observed that across different species, the cortex either organizes into columns or shows a dispersed distribution. This is for instance the case of the primary visual areas: mice and rats for instance present no specific organization of the neurons responsible for orientation selectivity (a so called salt-and-pepper structure), whereas the cat and primate present structured orientation selectivity columns (see e.g.~\cite{van-hooser-heimel-etal:05}). The reason of that difference is still poorly understood, and in that article, the authors suggest that the lack of orientation maps in rodent is related to the small size of V1 in these species. They argue that making a synaptic connection is easier locally. This possibility of creating synaptic connections is not a definitive answer to that question since the authors report a few counter-examples such as the ferret or the tree shrew. Including the delay distribution in this picture and how the topology and size of the cortex influence the synchronization property at the level of cortical columns could allow going deeper into the analysis of the reasons of the functional reliability of orientation selectivity across different species with sensitively different cortex size.

This work brings another refinement into the theoretical implication of delays in the dynamics of the network. The equations derived here provide a model of large-scale neuronal network with heterogeneous delay, and a preliminary study of the dynamics of such networks was initiated. Enriching this model by considering several populations is a straightforward extension and the present manuscript, and analyzing such models would allow going deeper into the analysis of neuronal networks and on the effect of heterogeneous delays in relationship with specific cortical functions. Eventually, the cortex is far from being the only system presenting heterogeneous distributions of interconnection delays. Most networks, including communication networks, internet, social networks and physical networks such as spin glass present delays in the communications, that depend on the specific distance between nodes. This study can hence be applied to these other cases, and applications of that framework to the TCP protocol is currently under investigation.

\begin{acknowledgements}
The author warmly thanks Tanguy Cabana for very interesting discussions and suggestions.
\end{acknowledgements}


\end{document}